\documentclass[12pt,oneside,peerreview,draftcls,onecolumn]{IEEEtran}
\usepackage{amsmath,amsfonts,amssymb,amsbsy, amsthm,ae,aecompl}
\usepackage{algorithm,algpseudocode}
\usepackage[utf8]{inputenc}
\usepackage[english]{babel}
\usepackage{bm}
\usepackage{color}
\usepackage[noadjust]{cite}
\usepackage{epsfig}
\usepackage{enumerate}
\usepackage{float}
\usepackage{fancyhdr}
\usepackage[T1]{fontenc}
\usepackage[acronym,toc,shortcuts]{glossaries}
\usepackage{graphicx, caption, subcaption}
\usepackage{hyperref}
\usepackage{lastpage}
\usepackage{listings}
\usepackage{lipsum}
\usepackage{dsfont}
\usepackage{multind}
\usepackage[normalem]{ulem}
\usepackage{transparent}
\usepackage{tikz,pgfplots}
\usepackage{times}
\usepackage{verbatim}
\usepackage[all]{xy}
\hypersetup{
    colorlinks,%
    citecolor=black,%
    filecolor=black,%
    linkcolor=black,%
    urlcolor=black
}

\usepackage{etoolbox}
\AtBeginEnvironment{align}{\setcounter{subeqn}{0}}
\newcounter{subeqn} \renewcommand{\thesubeqn}{\theequation\alph{subeqn}}%
\newcommand{\subeqn}{%
  \refstepcounter{subeqn}
  \tag{\thesubeqn}
}

\usetikzlibrary{calc}
\pgfplotsset{
  grid style = {
    dash pattern = on 0.025mm off 0.95mm on 0.025mm off 0mm, 
    line cap = round,
    black,
    line width = 0.5pt
  },
  tick label style={font=\small},
  label style={font=\small},
  legend style={font=\footnotesize},
}

\pgfplotscreateplotcyclelist{thechair4stocache}{
cyan!60!black,		solid,\\%
cyan!60!black, 		solid, every mark/.append style={solid,fill=cyan!60!black},mark=o\\%
lime!80!black,		densely dashed,\\%
lime!80!black, 		densely dashed, every mark/.append style={solid, fill=lime!80!black},mark=square\\%
red,                densely dotted,\\%
red,           		densely dotted, every mark/.append style={solid,fill=red},mark=triangle\\%
}

\pgfplotscreateplotcyclelist{thechair4stocache2}{
cyan!60!black,		solid,\\%
lime!80!black,		densely dashed,\\%
red,                densely dotted,\\%
cyan!60!black, 		solid, every mark/.append style={solid,fill=cyan!60!black},mark=o\\%
lime!80!black, 		densely dashed, every mark/.append style={solid, fill=lime!80!black},mark=square\\%
red,           		densely dotted, every mark/.append style={solid,fill=red},mark=triangle\\%
}

\graphicspath{{figures/}}

\bgroup
\makeglossaries
\newacronym{BS}{BS}{base station}
\newacronym{CDN}{CDN}{content delivery network}
\newacronym{CF}{CF}{collaborative filtering}
\newacronym{CRP}{CRP}{{C}hinese restaurant process}
\newacronym{CS}{CS}{central scheduler}
\newacronym{D2D}{D2D}{device-to-device}
\newacronym{HetNet}{HetNet}{heterogeneous network}
\newacronym{ICIC}{ICIC}{inter-cell interference coordination}
\newacronym{ICN}{ICN}{information-centric network}
\newacronym{LTE}{LTE}{long term evolution}
\newacronym{PPP}{PPP}{{P}oisson point process}
\newacronym{PHY}{PHY}{physical layer}
\newacronym{SBS}{SBS}{small base station}
\newacronym{SINR}{SINR}{signal-to-interference-plus-noise ratio}
\newacronym{SCN}{SCN}{small cell network}
\newacronym{SVD}{SVD}{singular value decomposition}
\newacronym{UT}{UT}{user terminal}
\newacronym{QoS}{QoS}{quality-of-service}
\newacronym{QoE}{QoE}{quality-of-experience}
\newacronym{RAN}{RAN}{radio access network}
\newacronym{PDF}{PDF}{probability distribution function}
\newacronym{PGFL}{PGFL}{probability generating functional}

\newtheorem{definition}{Definition}
\newtheorem{theorem}{Theorem}
\newtheorem{assumption}{Assumption}
\newtheorem{proposition}{Proposition}

\begin{document}
\title{Cache-enabled Small Cell Networks: Modeling  and Tradeoffs}
\author{ 
		\IEEEauthorblockN{Ejder Baştuğ$^{\diamond}$, Mehdi Bennis$^{\star}$, Marios Kountouris$^{\diamond, \dagger}$ and Mérouane Debbah$^{\diamond}$,}\\
		\IEEEauthorblockA{
				$^{\diamond}$Large Networks and Systems Group (LANEAS), CentraleSupélec, Gif-sur-Yvette, France \\	
				$^{\star}$Centre for Wireless Communications, University of Oulu, Finland \\
				$^{\dagger}$Telecommunication Department, CentraleSupélec, Gif-sur-Yvette, France \\	
				\{ejder.bastug, marios.kountouris, merouane.debbah\}@supelec.fr, bennis@ee.oulu.fi
		}
		\thanks{This research has been supported by the ERC Starting Grant 305123 MORE (Advanced Mathematical Tools for Complex Network Engineering), the SHARING project under the Finland grant 128010 and the project BESTCOM.}
}
\maketitle
\vspace{-100ex}
\IEEEpeerreviewmaketitle
 
\begin{abstract}
We consider a network model where small base stations (SBSs) have caching capabilities as a means to alleviate the backhaul load and satisfy users' demand. The SBSs are stochastically distributed over the plane according to a Poisson point process (PPP), and serve their users either (i) by bringing the content from the Internet through a finite rate backhaul or (ii) by serving them from the local caches. We derive closed-form expressions for the outage probability and the average delivery rate as a function of the signal-to-interference-plus-noise ratio (SINR), SBS density, target file bitrate, storage size, file length and file popularity. We then analyze the impact of key operating parameters on the system performance. It is shown that a certain outage probability can be achieved either by increasing the number of base stations or the total storage size. Our results and analysis provide key insights into the deployment of cache-enabled small cell networks (SCNs), which are seen as a promising solution for future heterogeneous cellular networks.
\end{abstract}

\section{Introduction}
Increasing traffic demand from mobile users due to the rich media applications, video streaming, social networks \cite{Cisco2014} is pushing mobile operators to make their mobile cellular networks evolving continuously (see \ac{LTE} \cite{3GPPRelease13}). \Glspl{SCN} \cite{Hoydis2011Green, Quek2013Small} and their integration with WiFi \cite{Mehdi2013When}, \glspl{HetNet} \cite{Andrews2013Seven}, together with many other ideas from both industry and academia, have now started being deployed and integrated in current cellular networks. In Europe, projects such as NewCom\# \cite{NewcomSharp} in the 7th Framework Program of the European Commission are focusing on the design of next generation cellular networks, and a new framework, called Horizon 2020 \cite{Horizon2020} is going to take place to support these efforts.

At the same time, content providers are moving their users' content to the intermediate nodes in the network, namely caching, yielding less delays for the access. \Glspl{CDN} such as Akamai \cite{Nygren2010Akamai} are for that purposes. In this context, \glspl{ICN} are emerging \cite{Ahlgren2012Survey}. Mixing these infrastructural concepts with cellular networks is also of interest \cite{Spagna2013TelcoCDN}\cite{Wanf2014CacheInTheAir}. Predicting users' behavior, and proactively caching the users' content in the edge of the network, namely base stations and user terminals, also shows that further gains can be obtained in terms of backhaul savings and user satisfaction \cite{Bastug2014LivingOnTheEdge}.

Even though the idea of caching in mobile cellular networks is somewhat recent, the origin of caching dates indeed back to 60s, where caching mechanisms are proposed to boost the performance of operating systems \cite{Belady1966Study}. Additionally, in past decades, many web caching schemes such as \cite{Borst2010Distributed} have appeared to sustain the data flow of the Internet. In the context of mobile cellular networks, there have been recent attempts on design of intelligent caching schemes by taking into account the wireless environment of mobile cellular networks. Due to its notorious non-tractability, these proposals are mainly based on approximate or heuristic solutions \cite{Bastug2013Proactive}\cite{Poularakis2014Multicast}\cite{Blasco2014LearningBased}. Beside these solutions, novel formulations and system models have been proposed to assess the performance of caching. For instance, information theoretical formulation of the caching problem is studied in \cite{Maddah2013Fundamental}. The expected cost of uncoded and coded data allocation strategies is given in \cite{Altman2013Coding}, where stochastically distributed cache-enabled nodes in a given area are assumed and the cost is defined as a function of distance. A game theoretical formulation of the caching problem as a many-to-many game is studied in \cite{Hamidouche2014Many} by taking into account data dissemination in social networks. The performance of caching in wireless \ac{D2D} networks is studied in \cite{Ji2014GridD2D} in a scenario where  nodes are placed on a grid and cache the content randomly. An alternative \ac{D2D} caching scenario with randomly located nodes is given in \cite{Altieri2014StoGeoD2D} and relevant tradeoffs curves are derived.  

The contribution of this work is to formulate the caching problem in a scenario where stochastically distributed \glspl{SBS} are equipped with storage units but have the limited backhaul capacity. In particular, we build on a tractable system model and define its performance metrics (outage probability and average delivery rate) as functions of  \ac{SINR}, number of \glspl{SBS}, target file bitrate, storage size, file length and file popularity distribution. By coupling the caching problem with \ac{PHY} in this way and relying on recent results from \cite{Andrews2011Tractable}, we show that a certain outage probability can be achieved either by 1) increasing number of \glspl{SBS} while the total storage size budged is fixed, or 2) increasing the total storage size while the number of \glspl{SBS} is fixed. To the best of our knowledge, our work differs from the aforementioned works in terms of studying deployment aspects of cache-enabled \glspl{SBS}. Similar line of work in terms of analysis with stochastic geometry tools can be found in \cite{Altieri2014StoGeoD2D, Altman2013Coding}. However, the system model and performance metrics are different than what is studied here.\footnote{Additionally, the related work \cite{Blaszczyszyn2014Geographic} was made public after the submission of this work.}

The rest of this paper is structured as follows. We describe our system model in  Section \ref{sec:systemmodel}. The performance metrics and main results are given in Section \ref{sec:permain}. In the same section, much simpler expressions are obtained by making specific assumptions on the system model. We validate these results via numerical simulations in Section \ref{sec:validation} and discuss the impact of parameters on the performance metrics. Then, a tradeoff between the number of deployed \glspl{SBS} and total storage size is given in Section \ref{sec:davidvsgoliath}. Finally, our conclusions and future perspectives are given in Section \ref{sec:conclusions}.\footnote{Compared to \cite{Bastug2014StoGeo}, this work contains more comprehensive mathematical treatment, proofs and the trade-off analysis conducted in Section \ref{sec:davidvsgoliath}.}
\section{System Model}
\label{sec:systemmodel}
The cellular network under consideration consists of \glspl{SBS}, whose locations are modeled according to a \ac{PPP} $\Phi$ with density $\lambda$. The broadband connection to these \glspl{SBS} is provided by a \ac{CS} via wired backhaul links. We assume that the broadband connection is finite and fixed, thus the backhaul link capacity of each \ac{SBS} is a decreasing function of $\lambda$. This in practice means that deploying more \glspl{SBS} in a certain area yields sharing the total broadband capacity among backhaul links. We will define this function more precisely in the next sections.

We suppose  that every \ac{SBS} has a storage unit with capacity $S$ nats (1 bit = $\text{ln}(2) = 0.693$ nats), thus they cache users' most popular files given in a catalog. The size of each file in the catalog has a length of $L$ nats and bitrate requirement of $T$ nats/sec/Hz. We note that the assumption on file length is for ease of analysis. Alternatively, the files in the catalog can be divided into chunks with the same length. The file popularity distribution of this catalog is a right continuous and monotonically decreasing \ac{PDF}, denoted as $f_{\mathrm{pop}}(f,\gamma)$. The parameter $f$ here corresponds to a point in the support of a file and $\gamma$ is the shape parameter of the distribution. We assume that this distribution is identical among all users. 

Every user equipped with a mobile user terminal is associated with the nearest \ac{SBS}, where its location falls into a point in a Poisson-Voronoi tessellation on the plane. In this model, we only consider the downlink transmission and overhead due to the file requests of users via uplink is neglected. In the downlink transmission, a tagged \ac{SBS} transmits with the constant transmit power $1/\mu$ Watts, and the standard unbounded power-law pathloss propagation model with exponent $\alpha > 2$  is used for the environment. The tagged \ac{SBS} and tagged user experience Rayleigh fading with mean $1$. Hence, the received power at the tagged user, located $r$-meters far away from its tagged \ac{SBS}, is given by $hr^{-\alpha}$. The random variable $h$ here follows an exponential distribution with mean $1/\mu$, represented as $h \sim \mathrm{Exponential}(\mu)$. 

Once users are associated with their closest \glspl{SBS}, we assume that they request some files (or chunks) randomly according to the file popularity distribution $f_{\mathrm{pop}}(f,\gamma)$. When requests reach to the \glspl{SBS} via uplink, the users are served immediately, either getting the file from the Internet via backhaul or being served from the local cache, depending on the availability of the file therein. If a requested file is available in the local cache of the \ac{SBS}, a \emph{cache hit} event occurs, otherwise a \emph{cache miss} event is said to be occurred. According to what we have explained so far, a sketch of the network model is given in Figure \ref{fig:systemmodel}.
\begin{figure}[!ht]	
	\centering
	\includegraphics[width=0.96\textwidth]{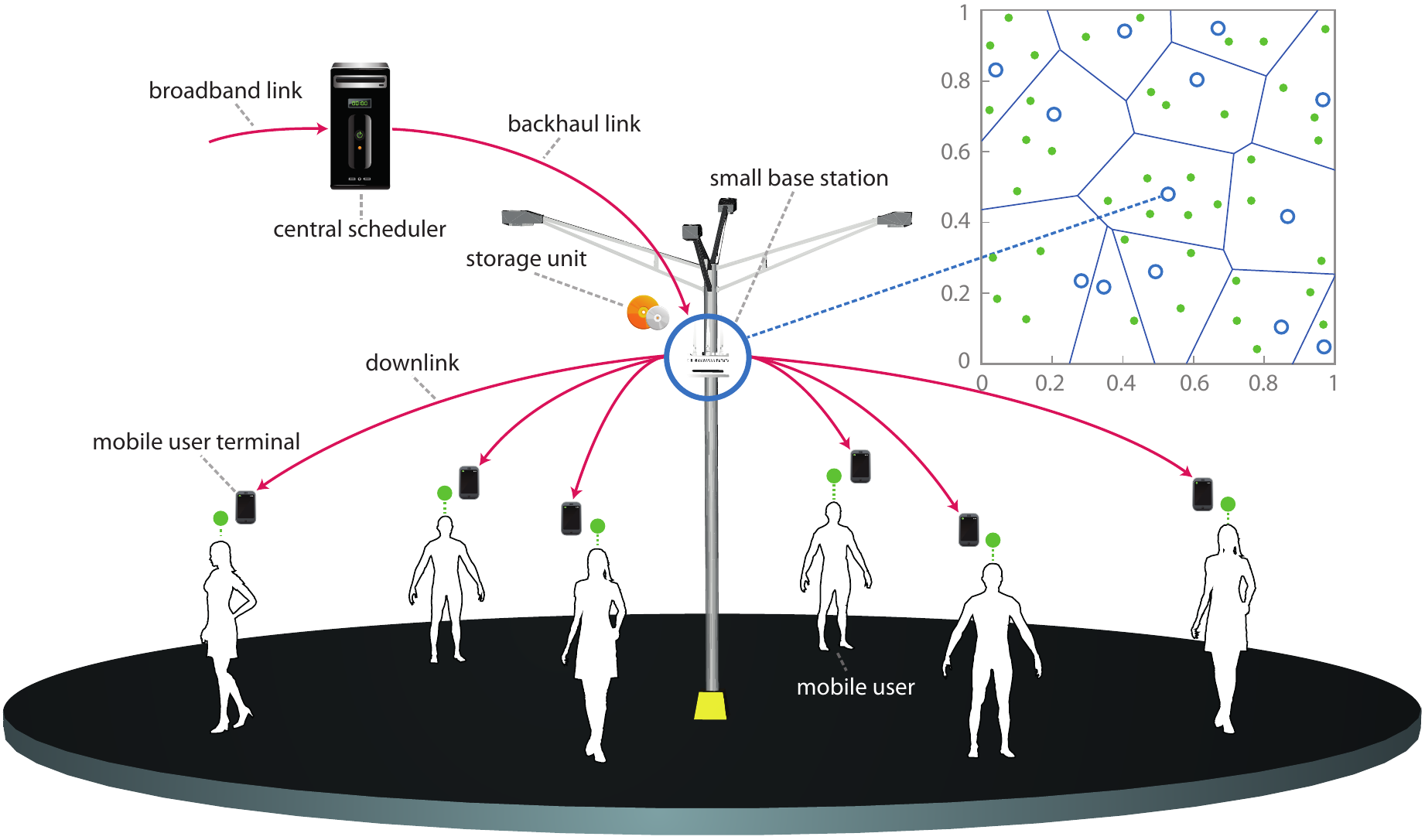}
	\caption{An illustration of the considered network model. The top right side of the figure shows a snapshot of \ac{PPP} per unit area where the \glspl{SBS} are randomly located. A closer look to communication structure of a cache-enabled \ac{SBS} is shown in the main figure.}
	\label{fig:systemmodel}
\end{figure}

In general, the performance of our system depends on several factors. To meet the \ac{QoE} requirements, the downlink rate provided to the requested user has to be equal or higher than the file bitrate $T$, so that the user does not observe any interruption during its experience. Although this requirement can be achieved in the downlink, yet another bottleneck can be the rate of the backhaul in case of cache misses. In the following, we define our performance metrics which take into account the aforementioned situations. We then present our main results in the same section.
\section{Performance Metrics and Main Results}
\label{sec:permain}
Performance metrics of interest in our system model are the \emph{outage probability} and \emph{average delivery rate}. We start by defining these metrics for the downlink. From now on, without loss of generality, we refer to the user $o$ as \emph{typical} user, which is located at the origin on the plane.

We know that the downlink rate depends on the \ac{SINR}. The \ac{SINR} of user $o$ which is located at a random distance $r$ far away from its \ac{SBS} $b_o$ is given by:
\begin{eqnarray}
	\label{eq:SINR}
	\textrm{SINR} 
	&\triangleq
	\frac{hr^{-\alpha}}{\sigma^2 + I_r},
\end{eqnarray}
where
\begin{eqnarray}
	\label{eq:I_r}
	I_r &\triangleq 
	\sum_{i \in \Phi / b_o}{g_i{R^{-\alpha}_i}},
\end{eqnarray}
is the total interference experienced from all other \glspl{SBS} except the connected \ac{SBS} $b_o$. Assume that the \emph{success probability} is the probability of the downlink rate exceeding the file bitrate $T$ and the probability of requested file being in the local cache. Then, the outage probability can be given as the complementary of the success probability as follows:
\begin{align}
	p_{\text{out}}(\lambda,T,\alpha,S, L, \gamma)  
	&\triangleq
	 1 - \underbrace{\mathbb{P}\Big[\mathrm{ln}(1 + \mathrm{SINR}) > T, f_o \in \Delta_{b_o}\Big]}_\text{success probability},
\end{align}
where $f_o$ is the requested file by the typical user, and $\Delta_{b_o}$ is the local cache of serving \ac{SBS} $b_o$. Indeed, such a definition of the outage probability comes from a simple observation. Ideally, if a requested file is in the cache of the serving \ac{SBS} (thus the limited backhaul is not used) and if the downlink rate is higher than the file bitrate $T$ (thus the user does not observe any interruption during the playback of the file), we then expect the outage probability to be close to zero. Given this explanation and the assumptions made in the previous section, we state the following theorem for outage probability. 
\begin{theorem}[Outage probability]
\label{the:outage-general}
The typical user has an outage probability from its tagged base station which can be expressed as:
\begin{multline}
	p_{\mathrm{out}}(\lambda,T,\alpha,S, L, \gamma)
    = 
	1 - 
	\pi\lambda			
	\int^{\infty}_{0}{}		\times \\
	{
	\int^{S/L}_{0}
			e^{-\pi\lambda v\beta(T,\alpha) - \mu(e^T - 1)\sigma^2v^{\alpha/2}}
			f_{\mathrm{pop}}(f,\gamma)
			\mathrm{d}f
			\mathrm{d}v
	},
\end{multline}
where $\beta(T,\alpha)$ is given by:
\begin{multline}
	\beta(T,\alpha) = \frac{2\left(\mu(e^T - 1)\right)}{\alpha} \times \\
	\mathbb{E}_g\left[
		g^{\frac{2}{\alpha}}
		\left(
			\Gamma\left(-\frac{2}{\alpha},\mu\left(e^T - 1\right)g\right) - 
			\Gamma\left(-\frac{2}{\alpha}\right)
		\right)							
	\right],
\end{multline}
where $\Gamma(a,x) = \int^{\infty}_{x}{t^{a-1}e^{-t}\mathrm{d}t}$ is the upper incomplete Gamma function and $\Gamma(x) = \int^{\infty}_{0}{t^{x-1}e^{-t}\mathrm{d}t}$ is the Gamma function.
\end{theorem}
\begin{proof}
	The proof is provided in Appendix \ref{app:outage-general}.
\end{proof}

Yet another useful metric in our system model is the delivery rate, which we define as follows:
\begin{align}
\label{eq:deliveryrate}
\tau 
&\triangleq
\begin{cases}
	T 			, & \text{if } \mathrm{ln}(1 + \mathrm{SINR}) > T \mathrm{\;and\;} f_o \in \Delta_{b_o}, \\
	C(\lambda) 	, & \text{if } \mathrm{ln}(1 + \mathrm{SINR}) > T \mathrm{\;and\;} f_o \not\in \Delta_{b_o}, \\
	0,			& \text{otherwise},
\end{cases}
\hspace{1.4cm}
\text{nats/sec/Hz}
\end{align}
where $C(\lambda)$ is the backhaul capacity provided to the \ac{SBS} for single frequency in the downlink.\footnote{Without loss of generality, more realistic values of delivery rate can be obtained by making a proper \ac{SINR} gap approximation and considering the total wireless bandwidth instead of $1$ Hz.} The definition above can be explained as follows. If the downlink rate is higher than the threshold $T$ (namely the bitrate of the requested file) and the requested file is available in the local cache, the rate $T$ is dedicated to the user by the tagged \ac{SBS}, which in turn is sufficient for \ac{QoE}. On the other hand, if the downlink rate is higher than $T$ but the requested file does not exist in the local cache of the tagged \ac{SBS}, the delivery rate will be limited by the backhaul link capacity $C(\lambda)$, for which we assume that $C(\lambda) < T$. Given this definition for the delivery rate, we state the following theorem.

\begin{theorem}[Average delivery rate]
\label{the:delivery-general}
The typical user has an average delivery rate from its tagged base station which can be expressed as:
\begin{multline}
	{\bar \tau}(\lambda,T,\alpha,S, L, \gamma) 
	= 
		\pi\lambda
			\int^{\infty}_{0}{
				e^{-\pi\lambda v\beta(T,\alpha) - \mu(e^T - 1)\sigma^2v^{\alpha/2}}\mathrm{d}v
		} \times \\
		\left(
			C(\lambda) +
			(T - C(\lambda))
			\int^{S/L}_{0}{
				f_{\mathrm{pop}}(f,\gamma)\mathrm{d}f
			}
		\right),
\end{multline}
where $\beta(T,\alpha)$ has the same definition as in Theorem \ref{the:outage-general}.
\end{theorem}
\begin{proof}
	The proof is deferred to Appendix \ref{app:delivery-general}.
\end{proof}
What we provided above are the general results. The exact values of outage probability and average delivery rate can be obtained by specifying the distribution of the interference, the backhaul link capacity $C(\lambda)$ and the file popularity distribution $f_{\mathrm{pop}}(f,\gamma)$. If this treatment does not yield closed form expressions, numerical integration can be done as a last resort for evaluating the functions. In the next section, as an example, we derive special cases of these results after some specific assumptions, which in turn yield much simpler expressions.
\subsection{Special Cases}
\label{sec:specialcases}
\begin{assumption}
\label{ass:special}
The following assumptions are given for the the system model:
\begin{enumerate}
	\item The noise power $\sigma^2$ is higher than $0$, and the pathloss component $\alpha$ is $4$.
	\item Interference is Rayleigh fading, which in turn $g_i \sim \mathrm{Exponential}(\mu)$.
	\item The capacity of backhaul links is given by: 
			\begin{equation}
				C\left(\lambda\right) 
				\triangleq
				\frac{C_1}{\lambda} + C_2,
			\end{equation}
		where $C_1 > 0$ and $C_2 \geq 0$ are some arbitrary coefficients such that $C(\lambda) < T$ holds.
	\item The file popularity distribution of users is characterized by a power law \cite{Newman2005Power} such as:
		\begin{align}
		f_{\mathrm{pop}}\left(f,\gamma\right)
		&\triangleq
		\begin{cases}
			\left(\gamma - 1\right)f^{-\gamma},
					& f \geq 1, \\
			0,			
				& f < 1,
		\end{cases}
		\end{align}
		where $\gamma > 1$ is the shape parameter of the distribution.
\end{enumerate}
\end{assumption}
The assumption $C(\lambda) < T$ comes from the observation that the high-speed fiber-optic backhaul links might be very costly in densely deployed \glspl{SBS} scenarios. Therefore, we assume that $C(\lambda)$ is lower than the bitrate of file. On the other hand, we characterize the file popularity distribution with a power law. Indeed, this  comes from the observation that many real world phenomena can be characterized by power laws (i.e. distribution of files in web proxies, distribution of word counts in natural languages) \cite{Newman2005Power}. According to our system model and the specific assumptions made in Assumption \ref{ass:special}, we state the following results.

\begin{proposition}[Outage probability]
	\label{the:outage-special}
	The typical user has an outage probability from its tagged base station which can be expressed as:
	\begin{multline}
	p_{\mathrm{out}}(\lambda,T,4,S, L, \gamma)  
	=
	1 -
	\frac{\pi^{\frac{3}{2}}\lambda}{\sqrt{\frac{e^T-1}{\mathrm{SNR}}}}
	\mathrm{exp}
		\left(
			\frac{\left(\lambda\pi(1 + \rho(T,4))\right)^2}{4(e^T-1)/\mathrm{SNR}}		\right) \times \\
	Q
		\left(
			\frac{\lambda\pi(1 + \rho(T,4))}{\sqrt{2(e^T-1)/\mathrm{SNR}}}
		\right)
	\left(1 - \left(\frac{L}{L+S}\right)^{\gamma - 1}\right),		
	\end{multline}
	where $\rho(T,4) = \sqrt{e^T - 1}\left(\frac{\pi}{2} - \mathrm{arctan}\left(\frac{1}{\sqrt{e^T-1}}\right) \right)$ and the standard Gaussian tail probability is given as $Q\left(x\right) = \frac{1}{\sqrt{2\pi}}\int_{x}^{\infty}{e^{-y^2/2}\mathrm{d}y}$.
\end{proposition}
\begin{proof}
	The proof is given in Appendix \ref{app:outage-special}.
\end{proof}

\begin{proposition}[Average delivery rate]
	\label{the:delivery-special}
	The typical user has an average delivery rate from its tagged base station which can be expressed as:
	\begin{multline}
		{\bar \tau}(\lambda,T,4,S, L, \gamma)  
		=
		\frac{\pi^{\frac{3}{2}}\lambda}{\sqrt{\frac{e^T-1}{\mathrm{SNR}}}}
		\mathrm{exp}
			\left(
				\frac{\left(\lambda\pi(1 + \rho(T,4))\right)^2}{4(e^T-1)/\mathrm{SNR}}
			\right) \times \\
		Q
			\left(
				\frac{\lambda\pi(1 + \rho(T,4))}{\sqrt{2(e^T-1)/\mathrm{SNR}}}
			\right) 
		\left(T +  \left(\frac{C_1}{\lambda} + C_2 - T\right)\left(\frac{L}{L+S}\right)^{\gamma - 1}\right),		
	\end{multline}
	where $\rho(T,4)$ and $Q\left(x\right)$ has the same definition as in Proposition \ref{the:outage-special}.
\end{proposition}
\begin{proof}
	The proof is given in Appendix \ref{app:delivery-special}.
\end{proof}
The expressions obtained for special cases are cumbersome but fairly easy to compute and does not require any integration. Note that $Q\left(x\right)$ function given in the expressions is a well-known function and can be computed by using lookup tables or standard numerical packages.
\section{Validation of the Proposed Model}
\label{sec:validation}
So far we have provided the results for outage probability and average delivery rate. In this section, we validate these results via Monte Carlo simulations. The numerical results shown here are obtained by averaging out over $1000$ realizations. In each realization, the \glspl{SBS} are distributed according to a \ac{PPP}. The file requests, signal and interfering powers of the typical user are drawn randomly according to the corresponding  probability distributions. The outage probability and average delivery rate are then calculated by considering \ac{SINR} and cache hit statistics. We note that all simulation curves match the theoretical ones. However, a slight mismatch is observed due to the fact that more precise discretization of continuous variables is avoided for affordable simulation times. As alluded to previously, the target file bit rate as well as average delivery rate are in units of nats/sec/Hz. On the other hand, the storage size and file lengths are in units of nats.
\subsection{Impact of storage size}
The storage size of \glspl{SBS} is one critical parameter in our system model. The effect of the storage size on the outage probability and the average delivery rate is plotted in Figures \ref{fig:plots-storage-outage} and \ref{fig:plots-storage-delivery}, respectively. Each curve represents a different value of target file bit rate. We observe that the outage probability reduces whereas the average delivery rate increases, as we increase the storage size. Such behavior, observed both in theoretical and simulation curves, confirms our initial intuition. 
\begin{figure}[!h]
\centering
\begin{tikzpicture}[scale=1.25, baseline]
	\begin{axis}[
		grid = major,
		cycle list name=thechair4stocache,
		legend cell align=left,
		legend style ={legend pos=north east},
		xlabel={Storage size $S$ [nats]},
		ylabel=Outage probability]
		
	\addplot+[semithick] plot coordinates {
		(0,	1)
		(0.1,	0.92071)
		(0.2,	0.85463)
		(0.3,	0.79872)
		(0.4,	0.7508)
		(0.5,	0.70927)
		(0.6,	0.67293)
		(0.7,	0.64086)
		(0.8,	0.61236)
		(0.9,	0.58685)
		(1,	0.5639)
		(1.1,	0.54314)
		(1.2,	0.52426)
		(1.3,	0.50702)
		(1.4,	0.49122)
		(1.5,	0.47668)
		(1.6,	0.46326)
		(1.7,	0.45084)
		(1.8,	0.4393)
		(1.9,	0.42856)
		(2,	0.41854)
		(2.1,	0.40916)
		(2.2,	0.40037)
		(2.3,	0.39211)
		(2.4,	0.38433)
		(2.5,	0.377)
		(2.6,	0.37008)
		(2.7,	0.36353)
		(2.8,	0.35733)
		(2.9,	0.35144)
		(3,	0.34585)
		(3.1,	0.34053)
		(3.2,	0.33547)
		(3.3,	0.33064)
		(3.4,	0.32603)
		(3.5,	0.32163)
		(3.6,	0.31741)
		(3.7,	0.31338)
		(3.8,	0.30951)
		(3.9,	0.3058)
		(4,	0.30224)
		(4.1,	0.29882)
		(4.2,	0.29553)
		(4.3,	0.29237)
		(4.4,	0.28932)
		(4.5,	0.28639)
		(4.6,	0.28355)
		(4.7,	0.28082)
		(4.8,	0.27818)
		(4.9,	0.27563)
		(5,	0.27317)
		(5.1,	0.27079)
		(5.2,	0.26848)
		(5.3,	0.26625)
		(5.4,	0.26408)
		(5.5,	0.26199)
		(5.6,	0.25995)
		(5.7,	0.25798)
		(5.8,	0.25607)
		(5.9,	0.25421)
		(6,	0.2524)
		(6.1,	0.25065)
		(6.2,	0.24894)
		(6.3,	0.24728)
		(6.4,	0.24567)
		(6.5,	0.2441)
		(6.6,	0.24257)
		(6.7,	0.24108)
		(6.8,	0.23962)
		(6.9,	0.23821)
		(7,	0.23683)
		(7.1,	0.23548)
		(7.2,	0.23417)
		(7.3,	0.23289)
		(7.4,	0.23164)
		(7.5,	0.23042)
		(7.6,	0.22922)
		(7.7,	0.22806)
		(7.8,	0.22692)
		(7.9,	0.2258)
		(8,	0.22471)
		(8.1,	0.22365)
		(8.2,	0.22261)
		(8.3,	0.22159)
		(8.4,	0.22059)
		(8.5,	0.21961)
		(8.6,	0.21866)
		(8.7,	0.21772)
		(8.8,	0.2168)
		(8.9,	0.2159)
		(9,	0.21502)
		(9.1,	0.21416)
		(9.2,	0.21331)
		(9.3,	0.21248)
		(9.4,	0.21167)
		(9.5,	0.21087)
		(9.6,	0.21009)
		(9.7,	0.20932)
		(9.8,	0.20856)
		(9.9,	0.20782)
		(10,	0.20709)
	}; 	\addlegendentry{$T = 0.1$ (The.)}

	\addplot+[only marks, semithick] plot coordinates {
		(0,	1)
		(1,	0.53984)
		(2,	0.40288)
		(3,	0.33602)
		(4,	0.29547)
		(5,	0.26872)
		(6,	0.24705)
		(7,	0.23245)
		(8,	0.22039)
		(9,	0.21183)
		(10,	0.20344)
	}; 	\addlegendentry{$T = 0.1$ (Sim.)}

	\addplot+[semithick] plot coordinates {
		(0,	1)
		(0.1,	0.92948)
		(0.2,	0.87072)
		(0.3,	0.821)
		(0.4,	0.77838)
		(0.5,	0.74144)
		(0.6,	0.70912)
		(0.7,	0.6806)
		(0.8,	0.65525)
		(0.9,	0.63257)
		(1,	0.61216)
		(1.1,	0.59369)
		(1.2,	0.5769)
		(1.3,	0.56157)
		(1.4,	0.54752)
		(1.5,	0.53459)
		(1.6,	0.52266)
		(1.7,	0.51161)
		(1.8,	0.50135)
		(1.9,	0.4918)
		(2,	0.48288)
		(2.1,	0.47454)
		(2.2,	0.46672)
		(2.3,	0.45938)
		(2.4,	0.45246)
		(2.5,	0.44594)
		(2.6,	0.43979)
		(2.7,	0.43397)
		(2.8,	0.42845)
		(2.9,	0.42321)
		(3,	0.41824)
		(3.1,	0.41351)
		(3.2,	0.40901)
		(3.3,	0.40471)
		(3.4,	0.40061)
		(3.5,	0.3967)
		(3.6,	0.39295)
		(3.7,	0.38936)
		(3.8,	0.38592)
		(3.9,	0.38262)
		(4,	0.37946)
		(4.1,	0.37642)
		(4.2,	0.37349)
		(4.3,	0.37068)
		(4.4,	0.36797)
		(4.5,	0.36536)
		(4.6,	0.36284)
		(4.7,	0.36041)
		(4.8,	0.35806)
		(4.9,	0.35579)
		(5,	0.3536)
		(5.1,	0.35148)
		(5.2,	0.34943)
		(5.3,	0.34745)
		(5.4,	0.34552)
		(5.5,	0.34366)
		(5.6,	0.34185)
		(5.7,	0.3401)
		(5.8,	0.33839)
		(5.9,	0.33674)
		(6,	0.33513)
		(6.1,	0.33357)
		(6.2,	0.33206)
		(6.3,	0.33058)
		(6.4,	0.32914)
		(6.5,	0.32775)
		(6.6,	0.32639)
		(6.7,	0.32506)
		(6.8,	0.32377)
		(6.9,	0.32251)
		(7,	0.32128)
		(7.1,	0.32009)
		(7.2,	0.31892)
		(7.3,	0.31778)
		(7.4,	0.31667)
		(7.5,	0.31558)
		(7.6,	0.31452)
		(7.7,	0.31348)
		(7.8,	0.31247)
		(7.9,	0.31148)
		(8,	0.31051)
		(8.1,	0.30956)
		(8.2,	0.30864)
		(8.3,	0.30773)
		(8.4,	0.30684)
		(8.5,	0.30597)
		(8.6,	0.30512)
		(8.7,	0.30429)
		(8.8,	0.30347)
		(8.9,	0.30267)
		(9,	0.30189)
		(9.1,	0.30112)
		(9.2,	0.30037)
		(9.3,	0.29963)
		(9.4,	0.29891)
		(9.5,	0.2982)
		(9.6,	0.2975)
		(9.7,	0.29682)
		(9.8,	0.29614)
		(9.9,	0.29549)
		(10,	0.29484)
	}; 	\addlegendentry{$T = 0.2$ (The.)}

	\addplot+[only marks, semithick] plot coordinates {
		(0,	1)
		(1,	0.5788)
		(2,	0.45344)
		(3,	0.39224)
		(4,	0.35512)
		(5,	0.33064)
		(6,	0.3108)
		(7,	0.29744)
		(8,	0.2864)
		(9,	0.27856)
		(10,	0.27088)
	}; 	\addlegendentry{$T = 0.2$ (Sim.)}

	\addplot+[semithick] plot coordinates {
		(0,	1)
		(0.1,	0.94229)
		(0.2,	0.89419)
		(0.3,	0.8535)
		(0.4,	0.81862)
		(0.5,	0.78839)
		(0.6,	0.76194)
		(0.7,	0.7386)
		(0.8,	0.71785)
		(0.9,	0.69929)
		(1,	0.68258)
		(1.1,	0.66747)
		(1.2,	0.65373)
		(1.3,	0.64118)
		(1.4,	0.62968)
		(1.5,	0.6191)
		(1.6,	0.60933)
		(1.7,	0.60029)
		(1.8,	0.59189)
		(1.9,	0.58408)
		(2,	0.57678)
		(2.1,	0.56995)
		(2.2,	0.56355)
		(2.3,	0.55754)
		(2.4,	0.55188)
		(2.5,	0.54655)
		(2.6,	0.54151)
		(2.7,	0.53674)
		(2.8,	0.53223)
		(2.9,	0.52794)
		(3,	0.52388)
		(3.1,	0.52)
		(3.2,	0.51632)
		(3.3,	0.5128)
		(3.4,	0.50945)
		(3.5,	0.50624)
		(3.6,	0.50317)
		(3.7,	0.50024)
		(3.8,	0.49742)
		(3.9,	0.49472)
		(4,	0.49213)
		(4.1,	0.48964)
		(4.2,	0.48725)
		(4.3,	0.48495)
		(4.4,	0.48273)
		(4.5,	0.48059)
		(4.6,	0.47853)
		(4.7,	0.47654)
		(4.8,	0.47462)
		(4.9,	0.47277)
		(5,	0.47097)
		(5.1,	0.46924)
		(5.2,	0.46756)
		(5.3,	0.46593)
		(5.4,	0.46436)
		(5.5,	0.46283)
		(5.6,	0.46135)
		(5.7,	0.45992)
		(5.8,	0.45852)
		(5.9,	0.45717)
		(6,	0.45586)
		(6.1,	0.45458)
		(6.2,	0.45334)
		(6.3,	0.45213)
		(6.4,	0.45096)
		(6.5,	0.44981)
		(6.6,	0.4487)
		(6.7,	0.44761)
		(6.8,	0.44656)
		(6.9,	0.44553)
		(7,	0.44452)
		(7.1,	0.44354)
		(7.2,	0.44259)
		(7.3,	0.44165)
		(7.4,	0.44074)
		(7.5,	0.43985)
		(7.6,	0.43898)
		(7.7,	0.43814)
		(7.8,	0.43731)
		(7.9,	0.4365)
		(8,	0.4357)
		(8.1,	0.43493)
		(8.2,	0.43417)
		(8.3,	0.43343)
		(8.4,	0.4327)
		(8.5,	0.43199)
		(8.6,	0.4313)
		(8.7,	0.43061)
		(8.8,	0.42995)
		(8.9,	0.42929)
		(9,	0.42865)
		(9.1,	0.42802)
		(9.2,	0.42741)
		(9.3,	0.4268)
		(9.4,	0.42621)
		(9.5,	0.42563)
		(9.6,	0.42506)
		(9.7,	0.4245)
		(9.8,	0.42395)
		(9.9,	0.42341)
		(10,	0.42288)
	}; 	\addlegendentry{$T = 0.4$ (The.)}

	\addplot+[only marks, semithick] plot coordinates {
		(0,	1)
		(1,	0.65093)
		(2,	0.54704)
		(3,	0.49632)
		(4,	0.46556)
		(5,	0.44527)
		(6,	0.42883)
		(7,	0.41775)
		(8,	0.4086)
		(9,	0.40211)
		(10,	0.39574)
	}; 	\addlegendentry{$T = 0.4$ (Sim.)}

	\end{axis}
\end{tikzpicture}
\caption{The evolution of outage probability with respect to the storage size. $\mathrm{SNR} = 10$ dB, $\lambda = 0.2$, $\gamma = 2$, $L = 1$ nats, $\alpha = 4$, $C_1 = 0.0005$, $C_2 = 0$.}
\label{fig:plots-storage-outage}
\vspace{0.35cm}
\end{figure}
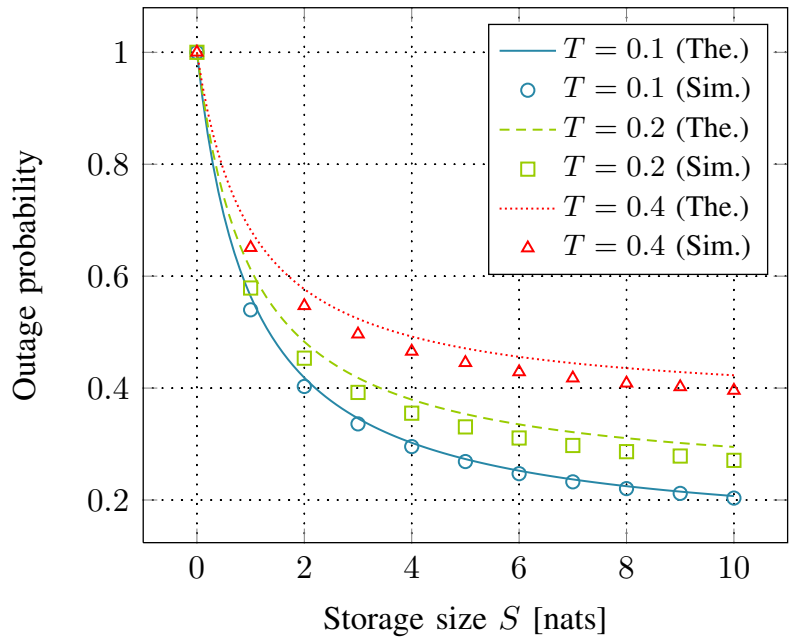
\begin{figure}[!h]
\centering
\begin{tikzpicture}[scale=1.25, baseline]
	\begin{axis}[
		grid = major,
		cycle list name=thechair4stocache,
		legend cell align=left,
		legend style ={legend pos=south east},
		xlabel={Storage size $S$ [nats]},
		ylabel={Avg. delivery rate [nats/sec/Hz]}]
		
	\addplot+[semithick] plot coordinates {
		(0,	0.0021805)
		(0.1,	0.0099113)
		(0.2,	0.016354)
		(0.3,	0.021805)
		(0.4,	0.026477)
		(0.5,	0.030527)
		(0.6,	0.03407)
		(0.7,	0.037197)
		(0.8,	0.039976)
		(0.9,	0.042462)
		(1,	0.0447)
		(1.1,	0.046725)
		(1.2,	0.048565)
		(1.3,	0.050246)
		(1.4,	0.051787)
		(1.5,	0.053204)
		(1.6,	0.054512)
		(1.7,	0.055724)
		(1.8,	0.056848)
		(1.9,	0.057896)
		(2,	0.058873)
		(2.1,	0.059788)
		(2.2,	0.060645)
		(2.3,	0.06145)
		(2.4,	0.062208)
		(2.5,	0.062923)
		(2.6,	0.063598)
		(2.7,	0.064236)
		(2.8,	0.064841)
		(2.9,	0.065415)
		(3,	0.06596)
		(3.1,	0.066478)
		(3.2,	0.066972)
		(3.3,	0.067443)
		(3.4,	0.067893)
		(3.5,	0.068322)
		(3.6,	0.068733)
		(3.7,	0.069126)
		(3.8,	0.069503)
		(3.9,	0.069865)
		(4,	0.070212)
		(4.1,	0.070545)
		(4.2,	0.070866)
		(4.3,	0.071174)
		(4.4,	0.071472)
		(4.5,	0.071758)
		(4.6,	0.072034)
		(4.7,	0.0723)
		(4.8,	0.072558)
		(4.9,	0.072806)
		(5,	0.073046)
		(5.1,	0.073279)
		(5.2,	0.073504)
		(5.3,	0.073721)
		(5.4,	0.073932)
		(5.5,	0.074137)
		(5.6,	0.074335)
		(5.7,	0.074527)
		(5.8,	0.074714)
		(5.9,	0.074895)
		(6,	0.075071)
		(6.1,	0.075242)
		(6.2,	0.075409)
		(6.3,	0.07557)
		(6.4,	0.075728)
		(6.5,	0.075881)
		(6.6,	0.07603)
		(6.7,	0.076176)
		(6.8,	0.076317)
		(6.9,	0.076455)
		(7,	0.07659)
		(7.1,	0.076721)
		(7.2,	0.076849)
		(7.3,	0.076974)
		(7.4,	0.077096)
		(7.5,	0.077215)
		(7.6,	0.077331)
		(7.7,	0.077445)
		(7.8,	0.077556)
		(7.9,	0.077665)
		(8,	0.077771)
		(8.1,	0.077875)
		(8.2,	0.077976)
		(8.3,	0.078076)
		(8.4,	0.078173)
		(8.5,	0.078268)
		(8.6,	0.078361)
		(8.7,	0.078453)
		(8.8,	0.078542)
		(8.9,	0.07863)
		(9,	0.078716)
		(9.1,	0.0788)
		(9.2,	0.078882)
		(9.3,	0.078963)
		(9.4,	0.079043)
		(9.5,	0.079121)
		(9.6,	0.079197)
		(9.7,	0.079272)
		(9.8,	0.079346)
		(9.9,	0.079418)
		(10,	0.079489)
	}; 	\addlegendentry{$T = 0.1$ (The.)}

	\addplot+[only marks,semithick] plot coordinates {
		(0,	0.002185)
		(1,	0.047051)
		(2,	0.060404)
		(3,	0.066923)
		(4,	0.070877)
		(5,	0.073484)
		(6,	0.075598)
		(7,	0.077021)
		(8,	0.078197)
		(9,	0.079032)
		(10,	0.07985)
	}; 	\addlegendentry{$T = 0.1$ (Sim.)}

	\addplot+[semithick] plot coordinates {
		(0,	0.0019392)
		(0.1,	0.015866)
		(0.2,	0.027472)
		(0.3,	0.037292)
		(0.4,	0.04571)
		(0.5,	0.053005)
		(0.6,	0.059388)
		(0.7,	0.06502)
		(0.8,	0.070026)
		(0.9,	0.074506)
		(1,	0.078537)
		(1.1,	0.082185)
		(1.2,	0.085501)
		(1.3,	0.088528)
		(1.4,	0.091304)
		(1.5,	0.093857)
		(1.6,	0.096214)
		(1.7,	0.098396)
		(1.8,	0.10042)
		(1.9,	0.10231)
		(2,	0.10407)
		(2.1,	0.10572)
		(2.2,	0.10726)
		(2.3,	0.10871)
		(2.4,	0.11008)
		(2.5,	0.11137)
		(2.6,	0.11258)
		(2.7,	0.11373)
		(2.8,	0.11482)
		(2.9,	0.11585)
		(3,	0.11684)
		(3.1,	0.11777)
		(3.2,	0.11866)
		(3.3,	0.11951)
		(3.4,	0.12032)
		(3.5,	0.12109)
		(3.6,	0.12183)
		(3.7,	0.12254)
		(3.8,	0.12322)
		(3.9,	0.12387)
		(4,	0.1245)
		(4.1,	0.1251)
		(4.2,	0.12567)
		(4.3,	0.12623)
		(4.4,	0.12677)
		(4.5,	0.12728)
		(4.6,	0.12778)
		(4.7,	0.12826)
		(4.8,	0.12872)
		(4.9,	0.12917)
		(5,	0.1296)
		(5.1,	0.13002)
		(5.2,	0.13043)
		(5.3,	0.13082)
		(5.4,	0.1312)
		(5.5,	0.13157)
		(5.6,	0.13192)
		(5.7,	0.13227)
		(5.8,	0.13261)
		(5.9,	0.13293)
		(6,	0.13325)
		(6.1,	0.13356)
		(6.2,	0.13386)
		(6.3,	0.13415)
		(6.4,	0.13443)
		(6.5,	0.13471)
		(6.6,	0.13498)
		(6.7,	0.13524)
		(6.8,	0.13549)
		(6.9,	0.13574)
		(7,	0.13599)
		(7.1,	0.13622)
		(7.2,	0.13645)
		(7.3,	0.13668)
		(7.4,	0.1369)
		(7.5,	0.13711)
		(7.6,	0.13732)
		(7.7,	0.13753)
		(7.8,	0.13773)
		(7.9,	0.13792)
		(8,	0.13811)
		(8.1,	0.1383)
		(8.2,	0.13848)
		(8.3,	0.13866)
		(8.4,	0.13884)
		(8.5,	0.13901)
		(8.6,	0.13918)
		(8.7,	0.13934)
		(8.8,	0.1395)
		(8.9,	0.13966)
		(9,	0.13982)
		(9.1,	0.13997)
		(9.2,	0.14012)
		(9.3,	0.14026)
		(9.4,	0.14041)
		(9.5,	0.14055)
		(9.6,	0.14068)
		(9.7,	0.14082)
		(9.8,	0.14095)
		(9.9,	0.14108)
		(10,	0.14121)
	}; 	\addlegendentry{$T = 0.2$ (The.)}

	\addplot+[only marks,semithick] plot coordinates {
		(0,	0.002)
		(1,	0.085187)
		(2,	0.10995)
		(3,	0.12203)
		(4,	0.12936)
		(5,	0.1342)
		(6,	0.13812)
		(7,	0.14076)
		(8,	0.14294)
		(9,	0.14448)
		(10,	0.146)
	}; 	\addlegendentry{$T = 0.2$ (Sim.)}

	\addplot+[semithick] plot coordinates {
		(0,	0.0015871)
		(0.1,	0.024528)
		(0.2,	0.043645)
		(0.3,	0.059821)
		(0.4,	0.073686)
		(0.5,	0.085702)
		(0.6,	0.096217)
		(0.7,	0.10549)
		(0.8,	0.11374)
		(0.9,	0.12112)
		(1,	0.12776)
		(1.1,	0.13377)
		(1.2,	0.13923)
		(1.3,	0.14422)
		(1.4,	0.14879)
		(1.5,	0.15299)
		(1.6,	0.15688)
		(1.7,	0.16047)
		(1.8,	0.16381)
		(1.9,	0.16692)
		(2,	0.16982)
		(2.1,	0.17253)
		(2.2,	0.17507)
		(2.3,	0.17746)
		(2.4,	0.17971)
		(2.5,	0.18183)
		(2.6,	0.18384)
		(2.7,	0.18573)
		(2.8,	0.18753)
		(2.9,	0.18923)
		(3,	0.19085)
		(3.1,	0.19239)
		(3.2,	0.19385)
		(3.3,	0.19525)
		(3.4,	0.19658)
		(3.5,	0.19786)
		(3.6,	0.19908)
		(3.7,	0.20024)
		(3.8,	0.20136)
		(3.9,	0.20243)
		(4,	0.20346)
		(4.1,	0.20445)
		(4.2,	0.20541)
		(4.3,	0.20632)
		(4.4,	0.2072)
		(4.5,	0.20805)
		(4.6,	0.20887)
		(4.7,	0.20966)
		(4.8,	0.21043)
		(4.9,	0.21116)
		(5,	0.21188)
		(5.1,	0.21256)
		(5.2,	0.21323)
		(5.3,	0.21388)
		(5.4,	0.2145)
		(5.5,	0.21511)
		(5.6,	0.2157)
		(5.7,	0.21627)
		(5.8,	0.21682)
		(5.9,	0.21736)
		(6,	0.21788)
		(6.1,	0.21839)
		(6.2,	0.21889)
		(6.3,	0.21937)
		(6.4,	0.21983)
		(6.5,	0.22029)
		(6.6,	0.22073)
		(6.7,	0.22116)
		(6.8,	0.22158)
		(6.9,	0.22199)
		(7,	0.22239)
		(7.1,	0.22278)
		(7.2,	0.22316)
		(7.3,	0.22353)
		(7.4,	0.22389)
		(7.5,	0.22425)
		(7.6,	0.22459)
		(7.7,	0.22493)
		(7.8,	0.22526)
		(7.9,	0.22558)
		(8,	0.22589)
		(8.1,	0.2262)
		(8.2,	0.2265)
		(8.3,	0.2268)
		(8.4,	0.22709)
		(8.5,	0.22737)
		(8.6,	0.22765)
		(8.7,	0.22792)
		(8.8,	0.22818)
		(8.9,	0.22844)
		(9,	0.2287)
		(9.1,	0.22895)
		(9.2,	0.22919)
		(9.3,	0.22943)
		(9.4,	0.22967)
		(9.5,	0.2299)
		(9.6,	0.23013)
		(9.7,	0.23035)
		(9.8,	0.23057)
		(9.9,	0.23078)
		(10,	0.23099)
	}; 	\addlegendentry{$T = 0.4$ (The.)}

	\addplot+[only marks,semithick] plot coordinates {
		(0,	0.0016575)
		(1,	0.14041)
		(2,	0.18171)
		(3,	0.20187)
		(4,	0.2141)
		(5,	0.22216)
		(6,	0.2287)
		(7,	0.2331)
		(8,	0.23674)
		(9,	0.23932)
		(10,	0.24185)
	}; 	\addlegendentry{$T = 0.4$ (Sim.)}

	\end{axis}
\end{tikzpicture}
\caption{The evolution of  average delivery rate with respect to the storage size. $\mathrm{SNR} = 10$ dB, $\lambda = 0.2$, $\gamma = 2$, $L = 1$ nats, $\alpha = 4$, $C_1 = 0.0005$, $C_2 = 0$.}
\label{fig:plots-storage-delivery}
\vspace{0.35cm}
\end{figure}
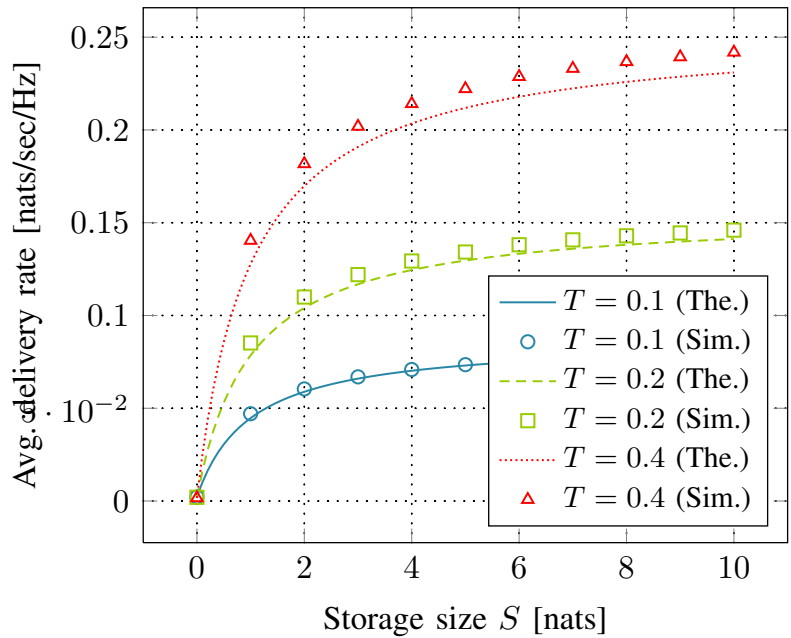
\subsection{Impact of the number of base stations}
The evolution of outage probability with respect to the number of base stations is depicted in Figure \ref{fig:plots-base}. As the base station density increases, the outage probability decreases. This decrement in outage probability can be improved further by increasing the storage size of \glspl{SBS}.
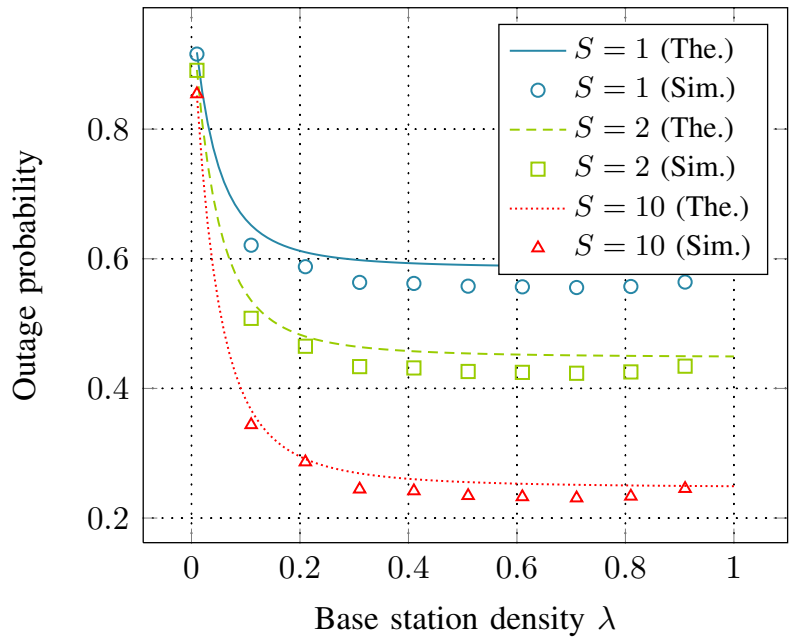
\begin{figure}[!h]
\centering
\begin{tikzpicture}[scale=1.25, baseline]
	\begin{axis}[
		grid = major,
		cycle list name=thechair4stocache,
		legend cell align=left,
		legend style ={legend pos=north east},
		xlabel={Base station density $\lambda$},
		ylabel=Outage probability]
		
	\addplot+[semithick] plot coordinates {
		(0.01,	0.91851)
		(0.02,	0.85653)
		(0.03,	0.80874)
		(0.04,	0.77143)
		(0.05,	0.74195)
		(0.06,	0.7184)
		(0.07,	0.69938)
		(0.08,	0.68388)
		(0.09,	0.67112)
		(0.1,	0.66054)
		(0.11,	0.65168)
		(0.12,	0.64421)
		(0.13,	0.63787)
		(0.14,	0.63246)
		(0.15,	0.6278)
		(0.16,	0.62378)
		(0.17,	0.62028)
		(0.18,	0.61722)
		(0.19,	0.61453)
		(0.2,	0.61216)
		(0.21,	0.61006)
		(0.22,	0.60819)
		(0.23,	0.60653)
		(0.24,	0.60503)
		(0.25,	0.60369)
		(0.26,	0.60248)
		(0.27,	0.60138)
		(0.28,	0.60039)
		(0.29,	0.59948)
		(0.3,	0.59865)
		(0.31,	0.5979)
		(0.32,	0.59721)
		(0.33,	0.59657)
		(0.34,	0.59598)
		(0.35,	0.59544)
		(0.36,	0.59494)
		(0.37,	0.59448)
		(0.38,	0.59405)
		(0.39,	0.59365)
		(0.4,	0.59328)
		(0.41,	0.59294)
		(0.42,	0.59261)
		(0.43,	0.59231)
		(0.44,	0.59203)
		(0.45,	0.59176)
		(0.46,	0.59151)
		(0.47,	0.59128)
		(0.48,	0.59106)
		(0.49,	0.59085)
		(0.5,	0.59066)
		(0.51,	0.59047)
		(0.52,	0.5903)
		(0.53,	0.59013)
		(0.54,	0.58997)
		(0.55,	0.58982)
		(0.56,	0.58968)
		(0.57,	0.58955)
		(0.58,	0.58942)
		(0.59,	0.5893)
		(0.6,	0.58919)
		(0.61,	0.58908)
		(0.62,	0.58897)
		(0.63,	0.58887)
		(0.64,	0.58878)
		(0.65,	0.58869)
		(0.66,	0.5886)
		(0.67,	0.58852)
		(0.68,	0.58844)
		(0.69,	0.58836)
		(0.7,	0.58829)
		(0.71,	0.58822)
		(0.72,	0.58815)
		(0.73,	0.58808)
		(0.74,	0.58802)
		(0.75,	0.58796)
		(0.76,	0.5879)
		(0.77,	0.58785)
		(0.78,	0.5878)
		(0.79,	0.58774)
		(0.8,	0.58769)
		(0.81,	0.58765)
		(0.82,	0.5876)
		(0.83,	0.58756)
		(0.84,	0.58751)
		(0.85,	0.58747)
		(0.86,	0.58743)
		(0.87,	0.58739)
		(0.88,	0.58736)
		(0.89,	0.58732)
		(0.9,	0.58729)
		(0.91,	0.58725)
		(0.92,	0.58722)
		(0.93,	0.58719)
		(0.94,	0.58716)
		(0.95,	0.58713)
		(0.96,	0.5871)
		(0.97,	0.58707)
		(0.98,	0.58705)
		(0.99,	0.58702)
		(1,	0.58699)
	}; 	\addlegendentry{$S = 1$ (The.)}

	\addplot+[only marks, semithick] plot coordinates {
		(0.01,	0.91576)
		(0.11,	0.62092)
		(0.21,	0.58775)
		(0.31,	0.56353)
		(0.41,	0.56195)
		(0.51,	0.55774)
		(0.61,	0.55669)
		(0.71,	0.55563)
		(0.81,	0.55721)
		(0.91,	0.56406)
	}; 	\addlegendentry{$S = 1$ (Sim.)}

	\addplot+[semithick] plot coordinates {
		(0.01,	0.89135)
		(0.02,	0.8087)
		(0.03,	0.74499)
		(0.04,	0.69524)
		(0.05,	0.65594)
		(0.06,	0.62453)
		(0.07,	0.59918)
		(0.08,	0.57851)
		(0.09,	0.5615)
		(0.1,	0.54738)
		(0.11,	0.53557)
		(0.12,	0.52562)
		(0.13,	0.51717)
		(0.14,	0.50995)
		(0.15,	0.50374)
		(0.16,	0.49837)
		(0.17,	0.4937)
		(0.18,	0.48962)
		(0.19,	0.48604)
		(0.2,	0.48288)
		(0.21,	0.48008)
		(0.22,	0.47759)
		(0.23,	0.47537)
		(0.24,	0.47338)
		(0.25,	0.47158)
		(0.26,	0.46997)
		(0.27,	0.46851)
		(0.28,	0.46718)
		(0.29,	0.46597)
		(0.3,	0.46487)
		(0.31,	0.46386)
		(0.32,	0.46294)
		(0.33,	0.46209)
		(0.34,	0.46131)
		(0.35,	0.46059)
		(0.36,	0.45993)
		(0.37,	0.45931)
		(0.38,	0.45874)
		(0.39,	0.45821)
		(0.4,	0.45771)
		(0.41,	0.45725)
		(0.42,	0.45682)
		(0.43,	0.45641)
		(0.44,	0.45604)
		(0.45,	0.45568)
		(0.46,	0.45535)
		(0.47,	0.45504)
		(0.48,	0.45474)
		(0.49,	0.45447)
		(0.5,	0.45421)
		(0.51,	0.45396)
		(0.52,	0.45373)
		(0.53,	0.45351)
		(0.54,	0.4533)
		(0.55,	0.4531)
		(0.56,	0.45291)
		(0.57,	0.45273)
		(0.58,	0.45256)
		(0.59,	0.4524)
		(0.6,	0.45225)
		(0.61,	0.4521)
		(0.62,	0.45196)
		(0.63,	0.45183)
		(0.64,	0.4517)
		(0.65,	0.45158)
		(0.66,	0.45146)
		(0.67,	0.45135)
		(0.68,	0.45125)
		(0.69,	0.45115)
		(0.7,	0.45105)
		(0.71,	0.45095)
		(0.72,	0.45086)
		(0.73,	0.45078)
		(0.74,	0.45069)
		(0.75,	0.45062)
		(0.76,	0.45054)
		(0.77,	0.45046)
		(0.78,	0.45039)
		(0.79,	0.45033)
		(0.8,	0.45026)
		(0.81,	0.4502)
		(0.82,	0.45014)
		(0.83,	0.45008)
		(0.84,	0.45002)
		(0.85,	0.44996)
		(0.86,	0.44991)
		(0.87,	0.44986)
		(0.88,	0.44981)
		(0.89,	0.44976)
		(0.9,	0.44972)
		(0.91,	0.44967)
		(0.92,	0.44963)
		(0.93,	0.44959)
		(0.94,	0.44954)
		(0.95,	0.44951)
		(0.96,	0.44947)
		(0.97,	0.44943)
		(0.98,	0.44939)
		(0.99,	0.44936)
		(1,	0.44932)
	}; 	\addlegendentry{$S = 2$ (The.)}

	\addplot+[only marks, semithick] plot coordinates {
		(0.01,	0.89069)
		(0.11,	0.5081)
		(0.21,	0.46505)
		(0.31,	0.43363)
		(0.41,	0.43158)
		(0.51,	0.42611)
		(0.61,	0.42475)
		(0.71,	0.42338)
		(0.81,	0.42543)
		(0.91,	0.43431)
	}; 	\addlegendentry{$S = 2$ (Sim.)}

	\addplot+[semithick] plot coordinates {
		(0.01,	0.85184)
		(0.02,	0.73914)
		(0.03,	0.65226)
		(0.04,	0.58442)
		(0.05,	0.53082)
		(0.06,	0.488)
		(0.07,	0.45343)
		(0.08,	0.42524)
		(0.09,	0.40204)
		(0.1,	0.38279)
		(0.11,	0.36669)
		(0.12,	0.35311)
		(0.13,	0.34159)
		(0.14,	0.33175)
		(0.15,	0.32328)
		(0.16,	0.31596)
		(0.17,	0.3096)
		(0.18,	0.30403)
		(0.19,	0.29915)
		(0.2,	0.29484)
		(0.21,	0.29102)
		(0.22,	0.28763)
		(0.23,	0.28459)
		(0.24,	0.28188)
		(0.25,	0.27943)
		(0.26,	0.27723)
		(0.27,	0.27524)
		(0.28,	0.27343)
		(0.29,	0.27178)
		(0.3,	0.27028)
		(0.31,	0.26891)
		(0.32,	0.26765)
		(0.33,	0.26649)
		(0.34,	0.26543)
		(0.35,	0.26444)
		(0.36,	0.26354)
		(0.37,	0.26269)
		(0.38,	0.26191)
		(0.39,	0.26119)
		(0.4,	0.26051)
		(0.41,	0.25988)
		(0.42,	0.2593)
		(0.43,	0.25875)
		(0.44,	0.25823)
		(0.45,	0.25775)
		(0.46,	0.2573)
		(0.47,	0.25687)
		(0.48,	0.25647)
		(0.49,	0.25609)
		(0.5,	0.25574)
		(0.51,	0.2554)
		(0.52,	0.25508)
		(0.53,	0.25478)
		(0.54,	0.2545)
		(0.55,	0.25423)
		(0.56,	0.25397)
		(0.57,	0.25373)
		(0.58,	0.2535)
		(0.59,	0.25328)
		(0.6,	0.25307)
		(0.61,	0.25287)
		(0.62,	0.25268)
		(0.63,	0.25249)
		(0.64,	0.25232)
		(0.65,	0.25216)
		(0.66,	0.252)
		(0.67,	0.25185)
		(0.68,	0.2517)
		(0.69,	0.25156)
		(0.7,	0.25143)
		(0.71,	0.2513)
		(0.72,	0.25118)
		(0.73,	0.25106)
		(0.74,	0.25095)
		(0.75,	0.25084)
		(0.76,	0.25073)
		(0.77,	0.25063)
		(0.78,	0.25054)
		(0.79,	0.25044)
		(0.8,	0.25035)
		(0.81,	0.25027)
		(0.82,	0.25018)
		(0.83,	0.2501)
		(0.84,	0.25003)
		(0.85,	0.24995)
		(0.86,	0.24988)
		(0.87,	0.24981)
		(0.88,	0.24974)
		(0.89,	0.24968)
		(0.9,	0.24961)
		(0.91,	0.24955)
		(0.92,	0.24949)
		(0.93,	0.24943)
		(0.94,	0.24938)
		(0.95,	0.24933)
		(0.96,	0.24927)
		(0.97,	0.24922)
		(0.98,	0.24917)
		(0.99,	0.24913)
		(1,	0.24908)
	}; 	\addlegendentry{$S = 10$ (The.)}

	\addplot+[only marks, semithick] plot coordinates {
		(0.01,	0.85418)
		(0.11,	0.34379)
		(0.21,	0.28637)
		(0.31,	0.24445)
		(0.41,	0.24172)
		(0.51,	0.23442)
		(0.61,	0.2326)
		(0.71,	0.23078)
		(0.81,	0.23351)
		(0.91,	0.24536)
	}; 	\addlegendentry{$S = 10$ (Sim.)}

	\end{axis}
\end{tikzpicture}
\caption{The evolution of outage probability with respect to the base station density. $\mathrm{SNR} = 10$ dB, $T = 0.2$, $\gamma = 2$, $L = 1$ nats, $\alpha = 4$, $C_1 = 0.0005$, $C_2 = 0$.}
\label{fig:plots-base}
\vspace{0.35cm}
\end{figure}
\subsection{Impact of target file bitrate}
Yet another important parameter in our setup is the target file bitrate $T$. Figure \ref{fig:plots-target} shows its impact on the outage probability for different values of storage size. Clearly, increasing the target file bitrate results in higher outage probability. However, this performance reduction can be compensated by increasing the storage size of \glspl{SBS}. The impact of storage size reduces, as $T$ increases.
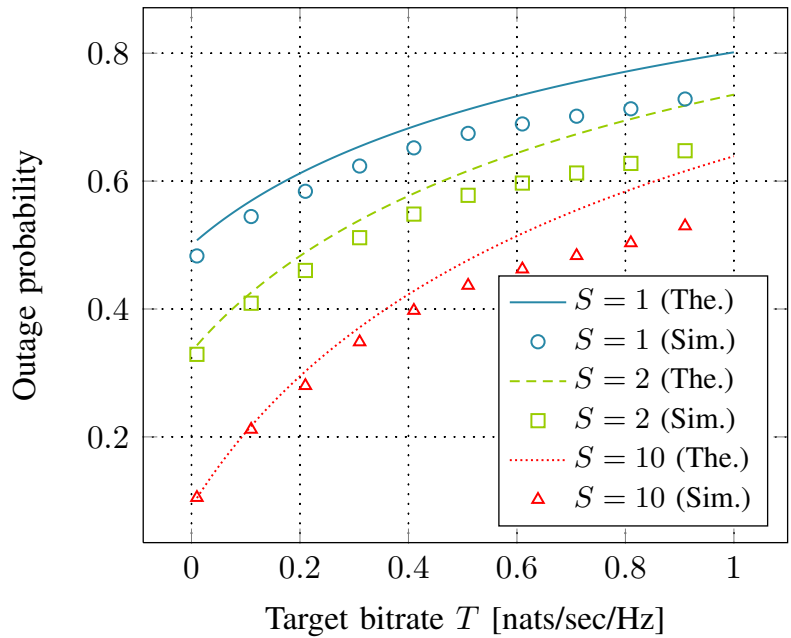
\begin{figure}[!h]
\centering
\begin{tikzpicture}[scale=1.25, baseline]
	\begin{axis}[
		grid = major,
		cycle list name=thechair4stocache,
		legend cell align=left,
		legend style ={legend pos=south east},
		xlabel={Target bitrate $T$ [nats/sec/Hz]},
		ylabel=Outage probability]
		
	\addplot+[semithick] plot coordinates {
		(0.01,	0.50739)
		(0.02,	0.51452)
		(0.03,	0.5214)
		(0.04,	0.52806)
		(0.05,	0.53449)
		(0.06,	0.54073)
		(0.07,	0.54678)
		(0.08,	0.55265)
		(0.09,	0.55836)
		(0.1,	0.5639)
		(0.11,	0.5693)
		(0.12,	0.57455)
		(0.13,	0.57966)
		(0.14,	0.58465)
		(0.15,	0.58951)
		(0.16,	0.59426)
		(0.17,	0.59889)
		(0.18,	0.60342)
		(0.19,	0.60784)
		(0.2,	0.61216)
		(0.21,	0.61639)
		(0.22,	0.62053)
		(0.23,	0.62457)
		(0.24,	0.62854)
		(0.25,	0.63242)
		(0.26,	0.63623)
		(0.27,	0.63995)
		(0.28,	0.64361)
		(0.29,	0.6472)
		(0.3,	0.65071)
		(0.31,	0.65416)
		(0.32,	0.65755)
		(0.33,	0.66088)
		(0.34,	0.66414)
		(0.35,	0.66735)
		(0.36,	0.6705)
		(0.37,	0.6736)
		(0.38,	0.67665)
		(0.39,	0.67964)
		(0.4,	0.68258)
		(0.41,	0.68548)
		(0.42,	0.68833)
		(0.43,	0.69113)
		(0.44,	0.69389)
		(0.45,	0.6966)
		(0.46,	0.69927)
		(0.47,	0.70191)
		(0.48,	0.7045)
		(0.49,	0.70705)
		(0.5,	0.70956)
		(0.51,	0.71204)
		(0.52,	0.71448)
		(0.53,	0.71689)
		(0.54,	0.71926)
		(0.55,	0.7216)
		(0.56,	0.7239)
		(0.57,	0.72618)
		(0.58,	0.72842)
		(0.59,	0.73063)
		(0.6,	0.73281)
		(0.61,	0.73496)
		(0.62,	0.73709)
		(0.63,	0.73918)
		(0.64,	0.74125)
		(0.65,	0.74329)
		(0.66,	0.74531)
		(0.67,	0.7473)
		(0.68,	0.74926)
		(0.69,	0.7512)
		(0.7,	0.75312)
		(0.71,	0.75501)
		(0.72,	0.75688)
		(0.73,	0.75873)
		(0.74,	0.76055)
		(0.75,	0.76236)
		(0.76,	0.76414)
		(0.77,	0.7659)
		(0.78,	0.76764)
		(0.79,	0.76936)
		(0.8,	0.77106)
		(0.81,	0.77274)
		(0.82,	0.7744)
		(0.83,	0.77604)
		(0.84,	0.77767)
		(0.85,	0.77927)
		(0.86,	0.78086)
		(0.87,	0.78243)
		(0.88,	0.78399)
		(0.89,	0.78552)
		(0.9,	0.78704)
		(0.91,	0.78855)
		(0.92,	0.79004)
		(0.93,	0.79151)
		(0.94,	0.79297)
		(0.95,	0.79441)
		(0.96,	0.79584)
		(0.97,	0.79725)
		(0.98,	0.79865)
		(0.99,	0.80003)
		(1,	0.8014)
	}; 	\addlegendentry{$S = 1$ (The.)}

	\addplot+[only marks, semithick] plot coordinates {
		(0.01,	0.48298)
		(0.11,	0.54458)
		(0.21,	0.58407)
		(0.31,	0.62355)
		(0.41,	0.65198)
		(0.51,	0.67462)
		(0.61,	0.68937)
		(0.71,	0.70147)
		(0.81,	0.71306)
		(0.91,	0.72833)
	}; 	\addlegendentry{$S = 1$ (Sim.)}

	\addplot+[semithick] plot coordinates {
		(0.01,	0.34319)
		(0.02,	0.3527)
		(0.03,	0.36187)
		(0.04,	0.37074)
		(0.05,	0.37933)
		(0.06,	0.38764)
		(0.07,	0.39571)
		(0.08,	0.40354)
		(0.09,	0.41114)
		(0.1,	0.41854)
		(0.11,	0.42573)
		(0.12,	0.43273)
		(0.13,	0.43955)
		(0.14,	0.4462)
		(0.15,	0.45268)
		(0.16,	0.45901)
		(0.17,	0.46519)
		(0.18,	0.47122)
		(0.19,	0.47712)
		(0.2,	0.48288)
		(0.21,	0.48852)
		(0.22,	0.49403)
		(0.23,	0.49943)
		(0.24,	0.50472)
		(0.25,	0.5099)
		(0.26,	0.51497)
		(0.27,	0.51994)
		(0.28,	0.52481)
		(0.29,	0.52959)
		(0.3,	0.53428)
		(0.31,	0.53889)
		(0.32,	0.5434)
		(0.33,	0.54784)
		(0.34,	0.55219)
		(0.35,	0.55647)
		(0.36,	0.56067)
		(0.37,	0.5648)
		(0.38,	0.56886)
		(0.39,	0.57285)
		(0.4,	0.57678)
		(0.41,	0.58064)
		(0.42,	0.58444)
		(0.43,	0.58817)
		(0.44,	0.59185)
		(0.45,	0.59547)
		(0.46,	0.59903)
		(0.47,	0.60254)
		(0.48,	0.606)
		(0.49,	0.6094)
		(0.5,	0.61275)
		(0.51,	0.61606)
		(0.52,	0.61931)
		(0.53,	0.62252)
		(0.54,	0.62568)
		(0.55,	0.6288)
		(0.56,	0.63187)
		(0.57,	0.6349)
		(0.58,	0.63789)
		(0.59,	0.64084)
		(0.6,	0.64375)
		(0.61,	0.64662)
		(0.62,	0.64945)
		(0.63,	0.65225)
		(0.64,	0.655)
		(0.65,	0.65773)
		(0.66,	0.66041)
		(0.67,	0.66307)
		(0.68,	0.66569)
		(0.69,	0.66827)
		(0.7,	0.67083)
		(0.71,	0.67335)
		(0.72,	0.67584)
		(0.73,	0.67831)
		(0.74,	0.68074)
		(0.75,	0.68314)
		(0.76,	0.68552)
		(0.77,	0.68786)
		(0.78,	0.69018)
		(0.79,	0.69248)
		(0.8,	0.69474)
		(0.81,	0.69698)
		(0.82,	0.6992)
		(0.83,	0.70139)
		(0.84,	0.70356)
		(0.85,	0.7057)
		(0.86,	0.70782)
		(0.87,	0.70991)
		(0.88,	0.71198)
		(0.89,	0.71403)
		(0.9,	0.71606)
		(0.91,	0.71807)
		(0.92,	0.72005)
		(0.93,	0.72201)
		(0.94,	0.72396)
		(0.95,	0.72588)
		(0.96,	0.72778)
		(0.97,	0.72967)
		(0.98,	0.73153)
		(0.99,	0.73338)
		(1,	0.7352)
	}; 	\addlegendentry{$S = 2$ (The.)}

	\addplot+[only marks, semithick] plot coordinates {
		(0.01,	0.3291)
		(0.11,	0.40903)
		(0.21,	0.46027)
		(0.31,	0.51151)
		(0.41,	0.5484)
		(0.51,	0.57778)
		(0.61,	0.59691)
		(0.71,	0.61263)
		(0.81,	0.62766)
		(0.91,	0.64747)
	}; 	\addlegendentry{$S = 2$ (Sim.)}

	\addplot+[semithick] plot coordinates {
		(0.01,	0.10435)
		(0.02,	0.11731)
		(0.03,	0.12982)
		(0.04,	0.14192)
		(0.05,	0.15363)
		(0.06,	0.16497)
		(0.07,	0.17597)
		(0.08,	0.18664)
		(0.09,	0.19701)
		(0.1,	0.20709)
		(0.11,	0.2169)
		(0.12,	0.22645)
		(0.13,	0.23575)
		(0.14,	0.24482)
		(0.15,	0.25366)
		(0.16,	0.26229)
		(0.17,	0.27071)
		(0.18,	0.27894)
		(0.19,	0.28698)
		(0.2,	0.29484)
		(0.21,	0.30253)
		(0.22,	0.31005)
		(0.23,	0.31741)
		(0.24,	0.32462)
		(0.25,	0.33168)
		(0.26,	0.33859)
		(0.27,	0.34537)
		(0.28,	0.35202)
		(0.29,	0.35854)
		(0.3,	0.36493)
		(0.31,	0.37121)
		(0.32,	0.37737)
		(0.33,	0.38341)
		(0.34,	0.38935)
		(0.35,	0.39519)
		(0.36,	0.40092)
		(0.37,	0.40655)
		(0.38,	0.41208)
		(0.39,	0.41753)
		(0.4,	0.42288)
		(0.41,	0.42814)
		(0.42,	0.43332)
		(0.43,	0.43842)
		(0.44,	0.44343)
		(0.45,	0.44837)
		(0.46,	0.45323)
		(0.47,	0.45801)
		(0.48,	0.46272)
		(0.49,	0.46736)
		(0.5,	0.47194)
		(0.51,	0.47644)
		(0.52,	0.48088)
		(0.53,	0.48525)
		(0.54,	0.48956)
		(0.55,	0.49382)
		(0.56,	0.49801)
		(0.57,	0.50214)
		(0.58,	0.50622)
		(0.59,	0.51024)
		(0.6,	0.5142)
		(0.61,	0.51812)
		(0.62,	0.52198)
		(0.63,	0.52579)
		(0.64,	0.52955)
		(0.65,	0.53326)
		(0.66,	0.53693)
		(0.67,	0.54054)
		(0.68,	0.54412)
		(0.69,	0.54764)
		(0.7,	0.55113)
		(0.71,	0.55457)
		(0.72,	0.55797)
		(0.73,	0.56133)
		(0.74,	0.56464)
		(0.75,	0.56792)
		(0.76,	0.57116)
		(0.77,	0.57436)
		(0.78,	0.57752)
		(0.79,	0.58065)
		(0.8,	0.58374)
		(0.81,	0.5868)
		(0.82,	0.58982)
		(0.83,	0.5928)
		(0.84,	0.59576)
		(0.85,	0.59868)
		(0.86,	0.60157)
		(0.87,	0.60442)
		(0.88,	0.60725)
		(0.89,	0.61004)
		(0.9,	0.61281)
		(0.91,	0.61554)
		(0.92,	0.61825)
		(0.93,	0.62093)
		(0.94,	0.62358)
		(0.95,	0.6262)
		(0.96,	0.6288)
		(0.97,	0.63137)
		(0.98,	0.63391)
		(0.99,	0.63642)
		(1,	0.63892)
	}; 	\addlegendentry{$S = 10$ (The.)}

	\addplot+[only marks, semithick] plot coordinates {
		(0.01,	0.10501)
		(0.11,	0.21164)
		(0.21,	0.27999)
		(0.31,	0.34835)
		(0.41,	0.39756)
		(0.51,	0.43675)
		(0.61,	0.46227)
		(0.71,	0.48324)
		(0.81,	0.50329)
		(0.91,	0.52972)
	}; 	\addlegendentry{$S = 10$ (Sim.)}

	\end{axis}
\end{tikzpicture}
\caption{The evolution of outage probability with respect to the target file bitrate. $\mathrm{SNR} = 10$ dB, $\lambda = 0.2$, $\gamma = 2$, $L = 1$ nats, $\alpha = 4$, $C_1 = 0.0005$, $C_2 = 0$.}
\label{fig:plots-target}
\vspace{0.35cm}
\end{figure}
\subsection{Impact of file popularity shape}
Another crucial parameter in our setup is the shape of the file popularity distribution, parameterized by $\gamma$. The impact of the parameter $\gamma$ on the outage probability, for different storage sizes, is given in Figure  \ref{fig:plots-popularity}. Generally, a higher value of $\gamma$ means that only a small portion of files is highly popular compared to the rest of the files. On the contrary, lower values of $\gamma$ correspond to a more uniform behavior on the popularity distribution. Therefore, as $\gamma$ increases, the outage probability reduces due to reduced requirement in terms of storage size. However, in very low and high values of $\gamma$, the impact on the outage probability is not high compared to the intermediate values.
\input{plots-popularity}
\section{David vs. Goliath: More \glspl{SBS} with less storage or less \glspl{SBS} with more storage?}
\label{sec:davidvsgoliath}
In the previous section, we have validated our results via numerical simulations and discussed the impact of several parameters on the outage probability and average delivery rate. On top of those, we are interested in finding a tradeoff between the \ac{SBS} density and the total storage size for a fixed set of parameters. We start by making an analogy with well-known David and Goliath story to examine the tradeoff between the \ac{SBS} density and total storage size.\footnote{David vs. Goliath refers to the underlying resource sharing problem which arises in a variety of scenarios including massive MIMO vs. Small Cells \cite{Hoydis2011David}.}
More precisely, we aim to answer the following question: Should we increase storage size of current \glspl{SBS} ({\bf David}) or deploy more \glspl{SBS} with less storage ({\bf Goliath}) in order to achieve a certain success probability? The answer is indeed useful for the realization of such a scenario. Putting more \glspl{SBS} in a given area may be not desirable due to increased deployment and operation costs ({\bf Evil}). Therefore, increasing the storage size of already deployed \glspl{SBS} may incur less cost ({\bf Good}). To characterize this tradeoff, we first define the optimal region as follows:
\begin{definition}[Optimal region] \label{def:optimalregion}
An outage probability $p^{\dagger}$ is said to be achievable if there exist some parameters $\lambda, T, \alpha, S, L, \gamma$ satisfying the following condition:
\begin{align}
	p_{\mathrm{out}}(\lambda,T,\alpha,S, L, \gamma) \leq p^{\dagger}.\notag
\end{align}
The set of all achievable $p^{\dagger}$ forms the optimal region.
\end{definition}
The optimal region can be tightened by restricting parameters $\lambda, T, \alpha, S, L, \gamma$ to some intervals. A detailed analysis on this is left for future work. Hereafter, we restrict ourselves to find the optimal \ac{SBS} density for a fixed set of parameters. In such a case, optimal \ac{SBS} density can be readily obtained by plugging these fixed parameters into $p_{\mathrm{out}}$ and solving the equation either analytically or numerically (i.e. bisection method \cite{Press2007Numerical}). In the following, we obtain a tradeoff curve between the \glspl{SBS} density and total storage size, by solving these equations systematically in the form of optimization problem.
\begin{definition}[\ac{SBS} density vs. total storage size tradeoff] Define the average total storage as $S_{\mathrm{total}} = {\lambda}S$, and fix $T$, $\alpha$, $L$ and $\gamma$ to some values in the optimal region given in Definition \ref{def:optimalregion}. Denote also $\lambda^{\star}$ as the optimal \ac{SBS} density for a given $S_{\mathrm{total}}$. Then, $\lambda^{\star}$  is  obtained by solving the following optimization problem:
\begin{align}
&	\underset{\lambda}{\mathrm{minimize}}	& & \lambda 		&	 \label{eq:tradeObjective} \\
&	\mathrm{subject}\text{ }\mathrm{to}	& & p_{\mathrm{out}}(\lambda,T,\alpha,S_{\mathrm{total}}/\lambda,L,\gamma) \leq p^{\dagger}. & 	\subeqn \label{eq:tradeC1Out}
\end{align}
The set of all achievable pairs $(\lambda^{\star}, S_{\mathrm{total}})$ characterize a tradeoff between the \ac{SBS} density and total storage size. 
\end{definition}
Figures \ref{fig:optimalDensity1} and \ref{fig:optimalDensity2} show two different configurations of the tradeoff. In these plots, to  achieve a certain outage probability (i.e. $p^{\dagger} = 0.3$), we see that it is sufficient to decrease the number of \glspl{SBS} by increasing the total storage size. Alternatively, the total storage size can be decreased by increasing the number of \glspl{SBS}. Moreover, for different values of parameter of interest (i.e. $T \in \{0.1, 0.2\}$ or $L \in \{1, 2\}$), there is also a scaling and shifting in this tradeoff. Regardless of this scaling and shifting, we see that David wins victory against Goliath.
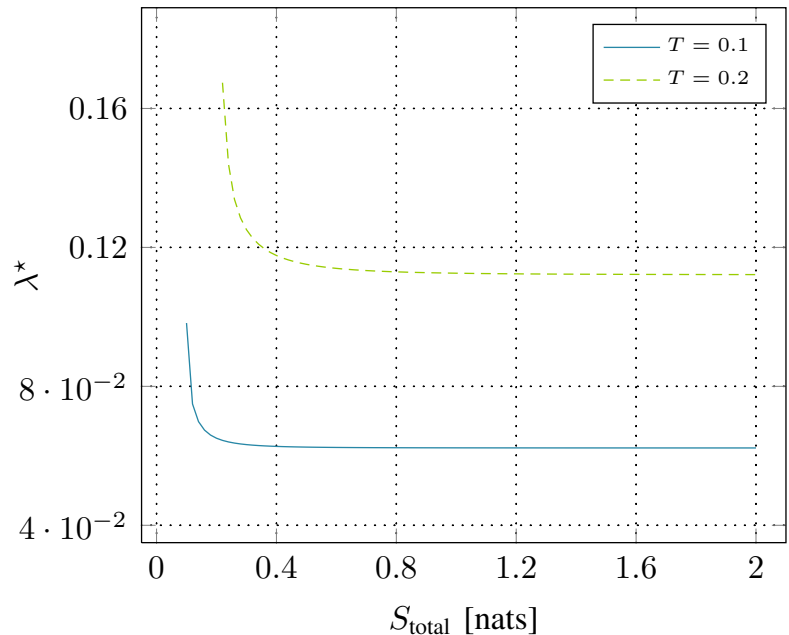
\begin{figure}[!h]
\centering
\begin{tikzpicture}[scale=1.25, baseline]
	\begin{axis}[
		grid = major,
		cycle list name=thechair4stocache2,
		mark repeat={10},
		ymin= 0.035,ymax=0.1890, ytick={0, 0.04, 0.08, 0.12, 0.16, 0.20},
		xmin=-0.05,xmax=2.10,   xtick={0, 0.40, 0.80, 1.20, 1.60, 2.00},
		legend cell align=left,
		legend style ={legend pos=north east, font=\tiny},
		xlabel= {$S_{\text{total}}$ [nats]},
		ylabel= $\lambda^{\star}$]
		
	\addplot plot coordinates {
		(0.1,	0.09823)
		(0.12,	0.074926)
		(0.14,	0.069846)
		(0.16,	0.067402)
		(0.18,	0.065978)
		(0.2,	0.065064)
		(0.22,	0.06444)
		(0.24,	0.063995)
		(0.26,	0.063667)
		(0.28,	0.063418)
		(0.3,	0.063226)
		(0.32,	0.063075)
		(0.34,	0.062955)
		(0.36,	0.062857)
		(0.38,	0.062777)
		(0.4,	0.062711)
		(0.42,	0.062655)
		(0.44,	0.062609)
		(0.46,	0.062569)
		(0.48,	0.062536)
		(0.5,	0.062507)
		(0.52,	0.062482)
		(0.54,	0.06246)
		(0.56,	0.062441)
		(0.58,	0.062424)
		(0.6,	0.062409)
		(0.62,	0.062396)
		(0.64,	0.062385)
		(0.66,	0.062374)
		(0.68,	0.062365)
		(0.7,	0.062357)
		(0.72,	0.062349)
		(0.74,	0.062343)
		(0.76,	0.062337)
		(0.78,	0.062331)
		(0.8,	0.062326)
		(0.82,	0.062322)
		(0.84,	0.062317)
		(0.86,	0.062314)
		(0.88,	0.06231)
		(0.9,	0.062307)
		(0.92,	0.062304)
		(0.94,	0.062301)
		(0.96,	0.062299)
		(0.98,	0.062296)
		(1,	0.062294)
		(1.02,	0.062292)
		(1.04,	0.062291)
		(1.06,	0.062289)
		(1.08,	0.062287)
		(1.1,	0.062286)
		(1.12,	0.062284)
		(1.14,	0.062283)
		(1.16,	0.062282)
		(1.18,	0.062281)
		(1.2,	0.06228)
		(1.22,	0.062279)
		(1.24,	0.062278)
		(1.26,	0.062277)
		(1.28,	0.062276)
		(1.3,	0.062275)
		(1.32,	0.062274)
		(1.34,	0.062274)
		(1.36,	0.062273)
		(1.38,	0.062272)
		(1.4,	0.062272)
		(1.42,	0.062271)
		(1.44,	0.062271)
		(1.46,	0.06227)
		(1.48,	0.06227)
		(1.5,	0.062269)
		(1.52,	0.062269)
		(1.54,	0.062269)
		(1.56,	0.062268)
		(1.58,	0.062268)
		(1.6,	0.062267)
		(1.62,	0.062267)
		(1.64,	0.062267)
		(1.66,	0.062266)
		(1.68,	0.062266)
		(1.7,	0.062266)
		(1.72,	0.062266)
		(1.74,	0.062265)
		(1.76,	0.062265)
		(1.78,	0.062265)
		(1.8,	0.062265)
		(1.82,	0.062264)
		(1.84,	0.062264)
		(1.86,	0.062264)
		(1.88,	0.062264)
		(1.9,	0.062264)
		(1.92,	0.062264)
		(1.94,	0.062263)
		(1.96,	0.062263)
		(1.98,	0.062263)
		(2,	0.062263)
	}; 	\addlegendentry{$T = 0.1$}

	\addplot plot coordinates {
		(0.22,	0.16741)
		(0.24,	0.144)
		(0.26,	0.13382)
		(0.28,	0.12852)
		(0.3,	0.12511)
		(0.32,	0.12273)
		(0.34,	0.12096)
		(0.36,	0.1196)
		(0.38,	0.11853)
		(0.4,	0.11767)
		(0.42,	0.11697)
		(0.44,	0.11639)
		(0.46,	0.1159)
		(0.48,	0.11549)
		(0.5,	0.11513)
		(0.52,	0.11483)
		(0.54,	0.11457)
		(0.56,	0.11434)
		(0.58,	0.11414)
		(0.6,	0.11396)
		(0.62,	0.1138)
		(0.64,	0.11367)
		(0.66,	0.11354)
		(0.68,	0.11343)
		(0.7,	0.11333)
		(0.72,	0.11324)
		(0.74,	0.11316)
		(0.76,	0.11309)
		(0.78,	0.11302)
		(0.8,	0.11296)
		(0.82,	0.11291)
		(0.84,	0.11285)
		(0.86,	0.11281)
		(0.88,	0.11277)
		(0.9,	0.11273)
		(0.92,	0.11269)
		(0.94,	0.11266)
		(0.96,	0.11263)
		(0.98,	0.1126)
		(1,	0.11257)
		(1.02,	0.11255)
		(1.04,	0.11252)
		(1.06,	0.1125)
		(1.08,	0.11248)
		(1.1,	0.11246)
		(1.12,	0.11245)
		(1.14,	0.11243)
		(1.16,	0.11242)
		(1.18,	0.1124)
		(1.2,	0.11239)
		(1.22,	0.11237)
		(1.24,	0.11236)
		(1.26,	0.11235)
		(1.28,	0.11234)
		(1.3,	0.11233)
		(1.32,	0.11232)
		(1.34,	0.11231)
		(1.36,	0.1123)
		(1.38,	0.1123)
		(1.4,	0.11229)
		(1.42,	0.11228)
		(1.44,	0.11228)
		(1.46,	0.11227)
		(1.48,	0.11226)
		(1.5,	0.11226)
		(1.52,	0.11225)
		(1.54,	0.11225)
		(1.56,	0.11224)
		(1.58,	0.11224)
		(1.6,	0.11223)
		(1.62,	0.11223)
		(1.64,	0.11222)
		(1.66,	0.11222)
		(1.68,	0.11222)
		(1.7,	0.11221)
		(1.72,	0.11221)
		(1.74,	0.11221)
		(1.76,	0.1122)
		(1.78,	0.1122)
		(1.8,	0.1122)
		(1.82,	0.11219)
		(1.84,	0.11219)
		(1.86,	0.11219)
		(1.88,	0.11219)
		(1.9,	0.11218)
		(1.92,	0.11218)
		(1.94,	0.11218)
		(1.96,	0.11218)
		(1.98,	0.11217)
		(2,	0.11217)
	}; 	\addlegendentry{$T = 0.2$}

	\end{axis}
\end{tikzpicture}
\caption{The trade-off between SBSs density and total storage size for different file target bitrates. $\mathrm{SNR} = 10$ dB, $\alpha = 4$, $L = 1$ nats, $\gamma = 3$ and $p^{\dagger} = 0.3$.}
\label{fig:optimalDensity1}
\vspace{0.35cm}
\end{figure}
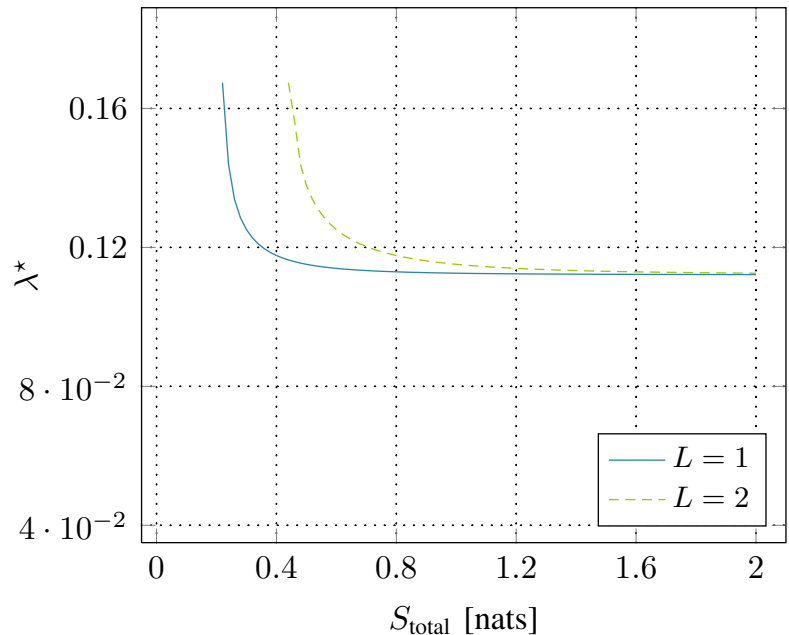
\begin{figure}[!h]
\centering
\begin{tikzpicture}[scale=1.25, baseline]
	\begin{axis}[
		grid = major,
		cycle list name=thechair4stocache2,
		mark repeat={10},
		ymin= 0.035,ymax=0.1890, ytick={0, 0.04, 0.08, 0.12, 0.16, 0.20},
		xmin=-0.05,xmax=2.10,   xtick={0, 0.40, 0.80, 1.20, 1.60, 2.00},		
		legend cell align=left,
		legend style ={legend pos=south east},
		xlabel= {$S_{\text{total}}$ [nats]},
		ylabel= $\lambda^{\star}$]
		
	\addplot plot coordinates {
		(0.22,	0.16741)
		(0.24,	0.144)
		(0.26,	0.13382)
		(0.28,	0.12852)
		(0.3,	0.12511)
		(0.32,	0.12273)
		(0.34,	0.12096)
		(0.36,	0.1196)
		(0.38,	0.11853)
		(0.4,	0.11767)
		(0.42,	0.11697)
		(0.44,	0.11639)
		(0.46,	0.1159)
		(0.48,	0.11549)
		(0.5,	0.11513)
		(0.52,	0.11483)
		(0.54,	0.11457)
		(0.56,	0.11434)
		(0.58,	0.11414)
		(0.6,	0.11396)
		(0.62,	0.1138)
		(0.64,	0.11367)
		(0.66,	0.11354)
		(0.68,	0.11343)
		(0.7,	0.11333)
		(0.72,	0.11324)
		(0.74,	0.11316)
		(0.76,	0.11309)
		(0.78,	0.11302)
		(0.8,	0.11296)
		(0.82,	0.11291)
		(0.84,	0.11285)
		(0.86,	0.11281)
		(0.88,	0.11277)
		(0.9,	0.11273)
		(0.92,	0.11269)
		(0.94,	0.11266)
		(0.96,	0.11263)
		(0.98,	0.1126)
		(1,	0.11257)
		(1.02,	0.11255)
		(1.04,	0.11252)
		(1.06,	0.1125)
		(1.08,	0.11248)
		(1.1,	0.11246)
		(1.12,	0.11245)
		(1.14,	0.11243)
		(1.16,	0.11242)
		(1.18,	0.1124)
		(1.2,	0.11239)
		(1.22,	0.11237)
		(1.24,	0.11236)
		(1.26,	0.11235)
		(1.28,	0.11234)
		(1.3,	0.11233)
		(1.32,	0.11232)
		(1.34,	0.11231)
		(1.36,	0.1123)
		(1.38,	0.1123)
		(1.4,	0.11229)
		(1.42,	0.11228)
		(1.44,	0.11228)
		(1.46,	0.11227)
		(1.48,	0.11226)
		(1.5,	0.11226)
		(1.52,	0.11225)
		(1.54,	0.11225)
		(1.56,	0.11224)
		(1.58,	0.11224)
		(1.6,	0.11223)
		(1.62,	0.11223)
		(1.64,	0.11222)
		(1.66,	0.11222)
		(1.68,	0.11222)
		(1.7,	0.11221)
		(1.72,	0.11221)
		(1.74,	0.11221)
		(1.76,	0.1122)
		(1.78,	0.1122)
		(1.8,	0.1122)
		(1.82,	0.11219)
		(1.84,	0.11219)
		(1.86,	0.11219)
		(1.88,	0.11219)
		(1.9,	0.11218)
		(1.92,	0.11218)
		(1.94,	0.11218)
		(1.96,	0.11218)
		(1.98,	0.11217)
		(2,	0.11217)
	}; 	\addlegendentry{$L = 1$}

	\addplot plot coordinates {
		(0.44,	0.16741)
		(0.46,	0.15662)
		(0.48,	0.144)
		(0.5,	0.13788)
		(0.52,	0.13382)
		(0.54,	0.13083)
		(0.56,	0.12852)
		(0.58,	0.12665)
		(0.6,	0.12511)
		(0.62,	0.12382)
		(0.64,	0.12273)
		(0.66,	0.12178)
		(0.68,	0.12096)
		(0.7,	0.12024)
		(0.72,	0.1196)
		(0.74,	0.11904)
		(0.76,	0.11853)
		(0.78,	0.11808)
		(0.8,	0.11767)
		(0.82,	0.1173)
		(0.84,	0.11697)
		(0.86,	0.11667)
		(0.88,	0.11639)
		(0.9,	0.11613)
		(0.92,	0.1159)
		(0.94,	0.11568)
		(0.96,	0.11549)
		(0.98,	0.1153)
		(1,	0.11513)
		(1.02,	0.11498)
		(1.04,	0.11483)
		(1.06,	0.11469)
		(1.08,	0.11457)
		(1.1,	0.11445)
		(1.12,	0.11434)
		(1.14,	0.11423)
		(1.16,	0.11414)
		(1.18,	0.11405)
		(1.2,	0.11396)
		(1.22,	0.11388)
		(1.24,	0.1138)
		(1.26,	0.11373)
		(1.28,	0.11367)
		(1.3,	0.1136)
		(1.32,	0.11354)
		(1.34,	0.11348)
		(1.36,	0.11343)
		(1.38,	0.11338)
		(1.4,	0.11333)
		(1.42,	0.11329)
		(1.44,	0.11324)
		(1.46,	0.1132)
		(1.48,	0.11316)
		(1.5,	0.11312)
		(1.52,	0.11309)
		(1.54,	0.11305)
		(1.56,	0.11302)
		(1.58,	0.11299)
		(1.6,	0.11296)
		(1.62,	0.11293)
		(1.64,	0.11291)
		(1.66,	0.11288)
		(1.68,	0.11285)
		(1.7,	0.11283)
		(1.72,	0.11281)
		(1.74,	0.11279)
		(1.76,	0.11277)
		(1.78,	0.11275)
		(1.8,	0.11273)
		(1.82,	0.11271)
		(1.84,	0.11269)
		(1.86,	0.11267)
		(1.88,	0.11266)
		(1.9,	0.11264)
		(1.92,	0.11263)
		(1.94,	0.11261)
		(1.96,	0.1126)
		(1.98,	0.11258)
		(2,	0.11257)
	}; 	\addlegendentry{$L = 2$}

	\end{axis}
\end{tikzpicture}
\caption{The trade-off between SBSs density and total storage size for different file lengths. $\mathrm{SNR} = 10$ dB, $\alpha = 4$, $T = 0.2$ nats/sec/Hz, $\gamma = 3$ and $p^{\dagger} = 0.3$.}
\label{fig:optimalDensity2}
\vspace{0.35cm}
\end{figure}
\section{Conclusions}
\label{sec:conclusions}
We have studied the caching problem in a scenario where \glspl{SBS} are stochastically distributed and have finite-rate backhaul links. We derived expressions for the outage probability and average delivery rate, and validate these results via numerical simulations. The results showed that significant gains in terms of outage probability and average delivery rate are possible by having cache-enabled \glspl{SBS}. We showed that telecom operators can either deploy more base stations or increase the storage size of existing deployment in order to achieve a certain \ac{QoE} level.
\bibliographystyle{IEEEtran}
\bibliography{references}
%
\appendices
\section{Proof of Theorem \ref{the:outage-general}}
\label{app:outage-general}
In order to prove Theorem \ref{the:outage-general}, we modify some useful results from \cite{Andrews2011Tractable}. Conditioning on the nearest base station at a distance $r$ from the typical user, the outage probability can be written as:
\begin{multline}
	p_{\text{out}}(\lambda,T,\alpha,S,\gamma)  
	= 
	\mathbb{E}_r\Big[
	 1 - \mathbb{P}[\mathrm{ln}(1 + \mathrm{SINR}) > T, f_o \in \Delta_{b_o} \mid r]
	 \Big]. \notag
\end{multline}
Since expectation is a linear operator and these two events are independent, the above expression can be decomposed as:
\begin{multline}
	\label{pro:outage-general:expanded}
	p_{\text{out}}(\lambda,T,\alpha,S,\gamma)  
	= 
	1 - 
	\underbrace{\mathbb{E}_r\Big[
	\mathbb{P}\left[\mathrm{ln}(1 + \mathrm{SINR}) > T \mid r \right]
	 \Big]}_\text{$(i)$}
	\underbrace{\mathbb{E}_r\Big[
	 \mathbb{P}\left[f_o \in \Delta_{b_o} \mid r\right]
	 \Big]}_\text{$(ii)$}.
\end{multline}
Proceeding term by term, we first write $(i)$ as:
\begin{flalign}
	& \mathbb{E}_r\left[\mathbb{P}\left[\mathrm{ln}(1 + \mathrm{SINR}) > T \mid r \right]\right] \notag &  
\end{flalign}
\vspace{-1.1cm}
\begin{align}
	&=
	\int_{r>0}{
		\mathbb{P}\left[\mathrm{ln}(1 + \mathrm{SINR}) > T \mid r \right]	f_r(r)\mathrm{d}r 
	} \label{pro:outage-general:firstterm-start} \\
	&\stackrel{(a)}{=} 
	\int_{r>0}{
		\mathbb{P}\left[\mathrm{ln}(1 + \mathrm{SINR}) > T \mid r \right]	e^{-\pi\lambda r^2}2\pi\lambda r\mathrm{d}r 
	} \notag \\
	&\stackrel{(b)}{=} 
	\int_{r>0}{
		\mathbb{P}\left[\frac{hr^{-\alpha}}{\sigma^2+I_r} > e^T-1 \mid r \right]	e^{-\pi\lambda r^2}2\pi\lambda r\mathrm{d}r
	} \notag \\
	&\stackrel{(c)}{=} 
	\int_{r>0}{
		\mathbb{P}\left[h > r^{\alpha}(e^T-1)(\sigma^2 + I_r) \mid r \right]	e^{-\pi\lambda r^2}2\pi\lambda r \mathrm{d}r
	}, \label{pro:outage-general:firstterm}
\end{align}
where $f_r(r) = e^{-\pi\lambda r^2}2\pi\lambda r$ is the \ac{PDF} of $r$ for \ac{PPP} \cite{Andrews2011Tractable}, hence $(a)$ follows from its substitution. The expression in $(b)$ is obtained by plugging the \ac{SINR} formula and letting it on the left hand side of the inequality, $(c)$ is the result of some algebraic manipulations for keeping fading variable $h$ alone. 

Conditioning on $I_r$ and using the fact that $h\sim \mathrm{Exponential(\mu)}$, the probability of random variable $h$ exceeding $r^{\alpha}(e^T-1)(\sigma^2 + I_r)$ can be written as:
\begin{flalign}
& \mathbb{P}\left[h > r^{\alpha}(e^T-1)(\sigma^2 + I_r) \mid r \right] \notag  &
\end{flalign}
\vspace{-1.1cm}
\begin{align}
	&= 
	\mathbb{E}_{I_r}\left[
			\mathbb{P}\left[h > r^{\alpha}(e^T-1)(\sigma^2 + I_r) \mid r, I_r \right]	
	\right] \notag \\
	&=
	\mathbb{E}_{I_r}\left[
			\mathrm{exp}\left(-\mu r^{\alpha}(e^T-1)(\sigma^2 + I_r)\right)
	\mid r \right] \notag \\
	&=
	e^{-\mu r^{\alpha}(e^T-1)\sigma^2}
	\mathcal{L}_{I_r}\left(\mu r^{\alpha}(e^T-1) \right), \label{pro:outage-general:h}
\end{align}
where $\mathcal{L}(s)$ is the Laplace transform of random variable $I_r$ evaluated at $s$ conditioned on the distance of the nearest base station from the origin. Substituting (\ref{pro:outage-general:h}) into  (\ref{pro:outage-general:firstterm}) yields the following:
\begin{multline}
	\label{pro:outage-general:firstterm2}
	\mathbb{E}_r\left[\mathbb{P}\left[\mathrm{ln}(1 + \mathrm{SINR}) > T \mid r \right]\right]
	= \\
	\int_{r>0}{
			e^{-\mu r^{\alpha}(e^T-1)\sigma^2}
			\mathcal{L}_{I_r}\left(\mu r^{\alpha}(e^T-1) \right)
				e^{-\pi\lambda r^2}2\pi\lambda r \mathrm{d}r
	}.
\end{multline}
Defining $g_i$ as a random variable of arbitrary but identical distribution for all $i$, and $R_i$ as the distance from the $i$-th base station to the tagged receiver, the Laplace transform is written as:
\begin{align}
	\mathcal{L}_{I_r}(s) &= \mathbb{E}_{I_r}\left[e^{-sI_r}\right] = 
		\mathbb{E}_{\Phi,\{g_i\}}
		\left[ 
			\mathrm{exp}\left( -s\sum_{i \in \Phi \backslash\{b_o\}}{g_iR^{-\alpha}_i}	\right)
		\right] \notag \\
	&=
		\mathbb{E}_{\Phi,\{g_i\}}
		\left[ 
			\prod_{i \in \Phi \backslash\{b_o\}}{ \mathrm{exp}\left(-sg_iR^{-\alpha}_i\right)}
		\right] \notag \\
	&\stackrel{(a)}{=} 
		\mathbb{E}_{\Phi}
		\left[ 
			\prod_{i \in \Phi \backslash\{b_o\}}{ \mathbb{E}_{\{g_i\}}\left[\mathrm{exp}\left(-sg_iR^{-\alpha}_i\right)\right]}
		\right] \notag \\
	&\stackrel{(b)}{=} 
		\mathbb{E}_{\Phi}
		\left[ 
			\prod_{i \in \Phi \backslash\{b_o\}}{ \mathbb{E}_g\left[\mathrm{exp}\left(-sgR^{-\alpha}_i\right)\right]}
		\right] \notag \\
	&=
	\mathrm{exp}\left(
		-2\pi\lambda
		\int^{\infty}_{r}{\left(1 -  \mathbb{E}_{g}
			\left[
				\mathrm{exp}\left(-sgv^{-\alpha}\right)
			\right]
		\right)v\mathrm{d}v}
	\right), \notag
\end{align}
where $(a)$ comes from the independence of $g_i$ from the point process $\Phi$, and $(b)$ follows from the i.i.d. assumption of $g_i$. The last step comes from the \ac{PGFL} of the \ac{PPP}, which basically says that for some function $f(x)$, $\mathbb{E}\left[\prod_{x \in \Phi}{f(x)}\right]=\mathrm{exp}\left(-\lambda\int_{\mathbb{R}^2}{(1 - f(x))\mathrm{d}x)} \right)$. Since the nearest interfering base station is at least at a distance $r$, the integration limits are from $r$ to infinity. Denoting $f(g)$ as the \ac{PDF} of $g$, then plugging in $s = \mu r^{\alpha}(e^T-1)$ and switching the integration order yields
\begin{multline}
	\mathcal{L}_{I_r}\left(\mu r^{\alpha}(e^T-1) \right) 
	= \\
		\mathrm{exp}\left(
		-2\pi\lambda
		\int^{\infty}_{0}{
		\left(
		\int^{\infty}_{r}{\left(1  -  e^{-\mu r^{\alpha}(e^T - 1)v^{-\alpha}g}
		\right)v\mathrm{d}v}
		\right)
		f(g)
		\mathrm{d}g
		}		
	\right). \nonumber
\end{multline}
By change of variables $v^{-\alpha} \rightarrow y$, the Laplace transform can be rewritten as:
\begin{multline}
	\label{pro:outage-general:laplace}
	\small
	\mathcal{L}_{I_r}\left(\mu r^{\alpha}(e^T-1) \right) 
	= \\
	\mathrm{exp}\Big(
		\lambda\pi r^2 - \frac{2\pi \lambda \left(\mu (e^T - 1)\right)^{\frac{2}{\alpha}}r^2}{\alpha}  \times \\
		\int^{\infty}_{0}{
			g^{\frac{2}{\alpha}}
			\left[
				\Gamma\left(-\frac{2}{\alpha},\mu\left(e^T - 1\right)g\right) -
				\Gamma\left(-\frac{2}{\alpha}\right)
			\right]	
			f(g)\mathrm{d}g	
		}
	\Big).
\end{multline}
Plugging (\ref{pro:outage-general:laplace}) into  (\ref{pro:outage-general:firstterm2}), using the substitution $r^2 \rightarrow v$ and after some algebraic manipulations, the expression becomes
\begin{multline}	
	\label{pro:outage-general:i}
	\mathbb{E}_r\left[\mathbb{P}\left[\mathrm{ln}(1 + \mathrm{SINR}) > T \mid r \right]\right] 
	= 
	\pi\lambda
	\int^{\infty}_{0}{
			e^{-\pi\lambda v\beta(T,\alpha) - \mu(e^T - 1)\sigma^2v^{\alpha/2}}\mathrm{d}v
	},
\end{multline}
where $\beta(T,\alpha)$ is given as
\begin{multline}
	\beta(T,\alpha) = \frac{2\left(\mu(e^T - 1)\right)}{\alpha}
	\mathbb{E}_g\left[
		g^{\frac{2}{\alpha}}
		\left(
			\Gamma\left(-\frac{2}{\alpha},\mu\left(e^T - 1\right)g\right) - 
			\Gamma\left(-\frac{2}{\alpha}\right)
		\right)							
	\right]. \notag
\end{multline}
So far, we have obtained $(i)$ of (\ref{pro:outage-general:expanded}). The term $(ii)$ is straightforward to derive. In the system model, as we assume that every small base station caches the same popular files and they have the same storage size, the cache hit probability becomes independent of the distance $r$. This yields: 
\begin{align}
	\label{pro:outage-general:ii}
	\mathbb{E}_r\left[\mathbb{P}\left[f_o \in \Delta_{b_o} \mid r\right]\right] 
	&= 
	\int^{S/L}_{0}{
			f_{\mathrm{pop}}(f,\gamma)\mathrm{d}f
	}.
\end{align}
Plugging both (\ref{pro:outage-general:i}) and (\ref{pro:outage-general:ii}) into (\ref{pro:outage-general:expanded}) and rearranging the terms, we conclude the proof. \hfill$\blacksquare$
\section{Proof of Theorem \ref{the:delivery-general}}
\label{app:delivery-general}
Average achievable delivery rate is ${\bar \tau} = \mathbb{E}\left[ \tau \right]$, where the average is taken over the \ac{PPP} and the fading distribution. It can be shown that
\begin{align}
	{\bar \tau} 
	&= \mathbb{E}\left[ \tau \right] \notag \\	
	&\stackrel{(a)}{=} 
		\mathbb{E}
		\Big[
			{\mathbb{P}\left[\mathrm{ln}(1 + \mathrm{SINR}) > T\right]}
			\Big(
				T\mathbb{P}\left[f_o \in \Delta_{b_o}\right]
				 + 
				C\left(\lambda \right)\mathbb{P}\left[f_o \not\in \Delta_{b_o}\right]
			\Big)
		\Big] \notag
\end{align}
\vspace{-1.0cm}
\begin{multline}			
	\hspace{0.20cm} \stackrel{(b)}{=} 
	 		\mathbb{E}\left[ \underbrace{\mathbb{P}\left[\mathrm{ln}(1 + \mathrm{SINR}) > T \mid r \right]}_\text{$\tau_1$} \right]	\times \\
	 		\left(
			\mathbb{E}\left[
				\underbrace{T\mathbb{P}\left[f_o \in \Delta_{b_o} \mid r \right]}_\text{$\tau_2$}
			\right]
			+ 
			\mathbb{E}\left[
				\underbrace{C\left(\lambda  \right)\mathbb{P}\left[f_o \not\in \Delta_{b_o} \mid r \right]}_\text{$\tau_3$}
			\right]
			\right) \notag
\end{multline}
\vspace{-1.0cm}
\begin{flalign}			
	\hspace{1.10cm} &=
		\mathbb{E}\left[\tau_1\right]
		\left(
		\mathbb{E}\left[\tau_2\right] +
		\mathbb{E}\left[\tau_3\right]
		\right), \label{pro:delivery-general:first} &
\end{flalign}
where $(a)$ is obtained by plugging the delivery rate as defined in (\ref{eq:deliveryrate}), and $(b)$ follows from independence of the events and linearity of the expectation operator.

Derivation of $\mathbb{E}[\tau_1]$ can be obtained from the proof of Theorem \ref{the:outage-general}, by following the steps from (\ref{pro:outage-general:firstterm-start}) to (\ref{pro:outage-general:i}). On the other hand, the fact that the cache hit probability is independent of $r$,  $\mathbb{E}_r[\tau_2]$ can be expressed as
\begin{align}
	\mathbb{E}_r[\tau_2] 
	&= T \int^{S/L}_{0}{
								f_{\mathrm{pop}}(f,\gamma)\mathrm{d}f
							}. \notag
\end{align} 
Using similar arguments, $\mathbb{E}_r[\tau_3]$ is written as:
\begin{align}
	\mathbb{E}_r[\tau_3] 
	&= C(\lambda)\left( 
								1 - 
								\int^{S/L}_{0}{
									f_{\mathrm{pop}}(f,\gamma)\mathrm{d}f
								}
							\right). \notag
\end{align}
Substituting these expressions into (\ref{pro:delivery-general:first}) concludes the proof.\hfill$\blacksquare$
\section{Proof of Proposition \ref{the:outage-special}}
\label{app:outage-special}
Since Proposition \ref{the:outage-special} is a special case of Theorem \ref{the:outage-general}, we follow the similar steps. We first rewrite (\ref{pro:outage-general:expanded}) as:
\begin{multline}
	\label{pro:outage-special:expanded}
	p_{\text{out}}(\lambda,T,\alpha,S,\gamma)  
	= 
	1 - 
	\underbrace{\mathbb{E}_r\Big[
	\mathbb{P}\left[\mathrm{ln}(1 + \mathrm{SINR}) > T \mid r \right]
	 \Big]}_\text{$(i)$}
	\underbrace{\mathbb{E}_r\Big[
	 \mathbb{P}\left[f_o \in \Delta_{b_o} \mid r\right]
	 \Big]}_\text{$(ii)$}.
\end{multline}
For the proceeding of $(i)$, the proof of Theorem \ref{the:outage-general} can be followed  starting from (\ref{pro:outage-general:firstterm-start}) to (\ref{pro:outage-general:firstterm2}). Then, the Laplace transform is written as
\begin{align}
	\mathcal{L}_{I_r}(s) 
	&=
		\mathbb{E}_{\Phi}
		\left[ 
			\prod_{i \in \Phi \backslash\{b_o\}}{ \mathbb{E}_{g}\left[\mathrm{exp}\left(-sgR^{-\alpha}_i\right)\right]}
		\right] \notag \\
	&\stackrel{(a)}{=} 
		\mathbb{E}_{\Phi}
		\left[ 
			\prod_{i \in \Phi \backslash\{b_o\}}{\frac{\mu}{\mu + sR_i^{-\alpha}}}
		\right] \notag \\
	&=
	\mathrm{exp}\left(
		-2\pi\lambda
		\int^{\infty}_{r}{\left(1 -  \frac{\mu}{\mu + sv^{-\alpha}}
		\right)v\mathrm{d}v}
	\right), \label{pro:outage-special:laplace0}
\end{align}
where $(a)$ comes from the new assumption that $g \sim \mathrm{Exponential}(\mu)$. Then, plugging $s = \mu r^{\alpha}\left( e^T - 1 \right)$ yields:
\begin{align}
	\mathcal{L}_{I_r}\left(\mu r^{\alpha}\left( e^T - 1 \right)\right) &=
	\mathrm{exp}\left(
		-2\pi\lambda
		\int^{\infty}_{r}{
			\frac{e^T - 1}{e^T - 1 + (\frac{v}{r})^{\alpha}}
			v\mathrm{d}v
		}
	\right). \notag
\end{align}
Using a change of variables $u = \left( \frac{v}{r(e^T-1)^{\alpha/2}}\right)^2$ results in
\begin{align}
	\label{pro:outage-special:laplace}
	\mathcal{L}_{I_r}\left(\mu r^{\alpha}\left( e^T - 1 \right)\right) 
	&=
	\mathrm{exp}\left(
		-\pi r^2\lambda\rho(T,\alpha)
	\right),
\end{align}
where
\begin{align}
	\rho(T,\alpha) 
	&= 
	(e^T - 1)^{2/\alpha}
	\int^{\infty}_{(e^T - 1)^{-2/\alpha}}{
		\frac{1}{1 + u^{\alpha/2}}
		\mathrm{d}u
	}. \notag
\end{align}
Substituting (\ref{pro:outage-special:laplace}) into (\ref{pro:outage-general:firstterm2}) with $r^2 \rightarrow v$ gives
\begin{align}
	\label{pro:outage-special:laplaceFinal}
	\pi\lambda\int_{0}^{\infty}{e^{-\pi \lambda v(1 + \rho(T,\alpha)) - \mu(e^T-1)\sigma^2v^{\alpha/2}}\mathrm{d}v}. 
\end{align}
Since $\alpha = 4$ in our special case, (\ref{pro:outage-special:laplaceFinal}) simplifies to
\begin{align}
	\label{pro:outage-special:alpha4}
	\pi\lambda\int_{0}^{\infty}{e^{-\pi \lambda v(1 + \rho(T,4)) - \mu(e^T-1)\sigma^2v^{2}}\mathrm{d}v},
\end{align}
where
\begin{align}
	\rho(T,4) 
	&= 
	(e^T - 1)^{2/\alpha}
	\int^{\infty}_{(e^T - 1)^{-2/\alpha}}{
		\frac{1}{1 + u^{2}}
		\mathrm{d}u
	} \notag \\
	&=
	(e^T - 1)^{2/\alpha}\left(\frac{\pi}{2} - \mathrm{arctan}\left( (e^T-1)^{-2/\alpha} \right) \right) \notag \\
	&=
	\sqrt{e^T - 1}\left(\frac{\pi}{2} - \mathrm{arctan}\left(\frac{1}{\sqrt{e^T-1}}\right) \right). \notag
\end{align}
From this point, (\ref{pro:outage-special:alpha4}) can be further simplified since it has a form similar to:
\begin{align}
	\int_{0}^{\infty}{e^{-ax}e^{-bx^2}\mathrm{d}x} 
	&= 
		\sqrt{\frac{\pi}{b}}
		\mathrm{exp}\left( \frac{a^2}{4b} \right)
		Q\left( \frac{a}{\sqrt{2b} }\right), \notag 
\end{align}
where $Q\left(x\right) = \frac{1}{\sqrt{2\pi}}\int_{x}^{\infty}{e^{-y^2/2}\mathrm{d}y}$ is the standard Gaussian tail probability. Setting $a = \pi\lambda(1 + \rho(T,4))$ and $b = \mu(e^T - 1)\sigma^2 = (e^T - 1)/\mathrm{SNR}$ gives
\begin{align}
	\label{pro:outage-special:i}
	\frac{\pi^{\frac{3}{2}}\lambda}{\sqrt{\frac{e^T-1}{\mathrm{SNR}}}}
	\mathrm{exp}
		\left(
			\frac{\left(\lambda\pi(1 + \rho(T,4))\right)^2}{4(e^T-1)/\mathrm{SNR}}
		\right)
	Q
		\left(
			\frac{\lambda\pi(1 + \rho(T,4))}{\sqrt{2(e^T-1)/\mathrm{SNR}}}
		\right).
\end{align}
This is the final expression for $(i)$ of (\ref{pro:outage-special:expanded}). The term $(ii)$ of (\ref{pro:outage-special:expanded}) can be obtained by using similar arguments given for (\ref{pro:outage-general:ii}) in the proof of Theorem \ref{the:outage-general}, meaning that the cache hit probability is independent of distance $r$. Thus:
\begin{align}
	\label{pro:outage-special:ii}
	\mathbb{E}_r\left[\mathbb{P}\left[f_o \in \Delta_{b_o} \mid r\right]\right] 
	&= 
	\int^{S/L}_{0}{
			f_{\mathrm{pop}}\left(f,\gamma\right)
			\mathrm{d}f
	} \notag \\
	&\stackrel{(a)}{=} 
	\int^{1 + S/L}_{1}{
			\left(\gamma - 1\right)f^{-\gamma}
			\mathrm{d}f
	} \notag \\
	&=
	1 - \left(\frac{L}{L+S}\right)^{\gamma - 1},	
\end{align}
where $(a)$ follows from plugging definition of $C(f,\lambda)$ given in Assumption \ref{ass:special} and changing the integration limits accordingly. The last term is the result of the integral. Therefore, we conclude the proof by plugging (\ref{pro:outage-special:i}) and (\ref{pro:outage-special:ii}) into (\ref{pro:outage-special:expanded}). \hfill$\blacksquare$

\section{Proof of Proposition \ref{the:delivery-special}}
\label{app:delivery-special}
The proposition is a special case of Theorem \ref{the:delivery-general}, thus we have the similar steps. We start by rewriting (\ref{pro:delivery-general:first}) as: 
\begin{multline}
	{\bar \tau} 
	=
	 		\mathbb{E}\left[ \underbrace{\mathbb{P}\left[\mathrm{ln}(1 + \mathrm{SINR}) > T \mid r \right]}_\text{$\tau_1$} \right]		\times \\
	 		\left(
			\mathbb{E}\left[
				\underbrace{T\mathbb{P}\left[f_o \in \Delta_{b_o} \mid r \right]}_\text{$\tau_2$}
			\right]
			+ 
			\mathbb{E}\left[
				\underbrace{C\left(\lambda  \right)\mathbb{P}\left[f_o \not\in \Delta_{b_o} \mid r \right]}_\text{$\tau_3$}
			\right]
			\right) \notag
\end{multline}
\begin{flalign}
		& \hspace{0.46cm} = \mathbb{E}\left[\tau_1\right]
		\left(
		\mathbb{E}\left[\tau_2\right] +
		\mathbb{E}\left[\tau_3\right]
		\right). \label{pro:delivery-special:expanded} &
\end{flalign}
In this expression, the term $\mathbb{E}\left[\tau_1\right]$ can be obtained from the proof of Proposition \ref{the:outage-special}. More precisely, observe that $\mathbb{E}\left[\tau_1\right]$ is identical to $(i)$ of (\ref{pro:outage-special:expanded}). Thus, following the steps from (\ref{pro:outage-special:laplace0}) to (\ref{pro:outage-special:i}), we obtain
\begin{align}
	\mathbb{E}\left[\tau_1\right]
	&=
	\mathbb{E}\Big[
		\mathbb{P}\left[\mathrm{ln}(1 + \mathrm{SINR}) > T \mid r \right]
	\Big] \notag \\
	&= 
	\frac{\pi^{\frac{3}{2}}\lambda}{\sqrt{\frac{e^T-1}{\mathrm{SNR}}}}
	\mathrm{exp}
		\left(
			\frac{\left(\lambda\pi(1 + \rho(T,4))\right)^2}{4(e^T-1)/\mathrm{SNR}}
		\right) 
	Q
		\left(
			\frac{\lambda\pi(1 + \rho(T,4))}{\sqrt{2(e^T-1)/\mathrm{SNR}}}
		\right). \label{pro:delivery-special:first}
\end{align}
On the other hand, $\mathbb{E}\left[\tau_2\right]$ can be obtained by taking $T$ out of the expectation and plugging (\ref{pro:outage-special:ii}) into the formula, i.e.
\begin{align}
	\mathbb{E}\left[\tau_2\right]
	&=
	\mathbb{E}\left[
		T\mathbb{P}\left[f_o \in \Delta_{b_o} \mid r \right]
	\right] \notag \\
	&= 
	T\left(1 - \left(\frac{L}{L+S}\right)^{\gamma - 1}\right).  \label{pro:delivery-special:second}
\end{align}
Finally, $\mathbb{E}\left[\tau_3\right]$ is easy to derive as 
\begin{align}
	\mathbb{E}\left[\tau_3\right]
	&=
	\mathbb{E}\left[
		C\left(\lambda  \right)\mathbb{P}\left[f_o \not\in \Delta_{b_o} \mid r \right]
	\right] \notag \\
	&=
	C\left(\lambda  \right)\left(\frac{L}{L+S}\right)^{\gamma - 1} \notag \\
	&=
		\left(\frac{C_1}{\lambda} + C_2\right)
		\left(\frac{L}{L+S}\right)^{\gamma - 1},  \label{pro:delivery-special:third}
\end{align}
where definition of $C(\lambda)$ follows from Assumption \ref{ass:special}. Substituting (\ref{pro:delivery-special:first}), (\ref{pro:delivery-special:second}) and (\ref{pro:delivery-special:third}) into (\ref{pro:delivery-special:expanded}) concludes the proof. \hfill$\blacksquare$
\end{document}